\documentclass[sn-mathphys-num]{sn-jnl}

\usepackage{graphicx}%
\usepackage{multirow}%
\usepackage{amsmath,amssymb,amsfonts}%
\usepackage{amsthm}%
\usepackage{mathrsfs}%
\usepackage[title]{appendix}%
\usepackage{xcolor}%
\usepackage{textcomp}%
\usepackage{manyfoot}%
\usepackage{booktabs}%
\usepackage{algorithm}%
\usepackage{algorithmicx}%
\usepackage{algpseudocode}%
\usepackage{listings}%

\theoremstyle{thmstyleone}%
\newtheorem{theorem}{Theorem}
\theoremstyle{thmstyletwo}%

\theoremstyle{thmstylethree}%

\usepackage[T1]{fontenc}
\usepackage{graphicx}
  \usepackage[misc]{ifsym}
  
\usepackage{mathtools}
\usepackage{oz}
\usepackage{enumerate}
\usepackage{array}

\newcommand*\rot{\rotatebox{90}}

\newcommand{\wpif}{wpif}

\newcommand{\goto}{\text{\sf goto}\;}

\newcommand{\ifc}{\text{\sf if}\;}
\newcommand{\elsec}{\;\text{\sf else}\;}
\newcommand{\return}{\text{\sf return}}
\newcommand{\flush}{\text{\sf flush}}
\newcommand{\wb}{\text{\sf wb}}
\newcommand{\eval}{\text{\sf eval}}
\newcommand{\mfence}{\text{\sf mfence}}
\newcommand{\RMW}{\text{\sf RMW}}

\newcommand{\LL}{{\cal L}}
\newcommand{\R}{{\cal R}}
\newcommand{\G}{{\cal G}}

\begin{document}

\title[Information flow security on persistent memory]{Information flow security on persistent memory}

\author{\fnm{Graeme} \sur{Smith}}

\affil{\orgdiv{School of Electrical Engineering and Computer Science}, \orgname{The University of Queensland}, 
\orgaddress{
\city{Brisbane}, 
\country{Australia}}}


\abstract{Persistent memory is a recently proposed memory paradigm that delivers many system-wide benefits, including improved runtime efficiency and the ability of programs to recover from power outages and system crashes. While recent research has investigated techniques for proving functional correctness of programs running on related architectures, this is not the case for the orthogonal concept of information flow security. In this paper, we provide an information flow logic for an unstructured language (i.e., with gotos rather than loops) modelling a simple assembly language. We apply this logic to x86 assembly using a notion of reordering interference freedom (rif) to reason about potential out-of-order propagation of instructions to memory. We then show how this same notion of rif can be used to similarly reason about information flow on persistent memory.}


\maketitle

\section{Introduction}
\label{sec:intro}

Traditionally, microprocessors use two tiers of memory: byte-addressable volatile memory (DRAM) for fast read and write access during program execution, and non-volatile hard, or solid-state, drive storage for longer-term retention of data. Unlike the data in volatile memory, that in non-volatile storage survives power shutdowns, including in the case of a system crash or power failure. However, accessing this data during program execution is less convenient, done at the level of blocks rather than bytes, and less efficient, being substantially slower than accessing volatile memory.

Persistent memory is a recently proposed memory paradigm that can be used as an additional layer between traditional volatile memory and non-volatile storage \cite{bal22}. Like volatile memory it allows fast, byte-addressable access. Unlike standard volatile memory, the data in persistent memory is preserved across power shutdowns, improving system performance on reboot. While there are systems where persistent memory can replace non-volatile storage completely, and also systems where persistent memory can replace volatile memory, we focus on its use as a layer between the traditional two tiers of memory. 


Writing programs to run on such persistent memory systems must be done with care. The order in which program writes are \emph{persisted}, i.e., committed to persistent memory from volatile memory, is largely under the control of the hardware based on, for example, the need for a cache line to be cleared. Therefore, unlike volatile memory which is simply repopulated from the hard/solid-state drive when a system is started, rebooting from persistent memory can leave the system in an inconsistent state. 

This has led to recent research on formal semantics and formal verification techniques for programs on persistent memory \cite{izr16,raa18,raa19,der19,raa20a,raa20b,bil20,cho21,der21,bil22a,bil22b,dos23}. These existing approaches focus on verifying functional correctness of such programs, and develop new verification techniques, e.g., based on buffering stores to memory \cite{raa20a,raa20b} or thread-observable views \cite{cho21,bil22a}. In contrast, in this paper we focus on information flow security on persistent memory, and build on the insight that the observable effects of persistent memory are comparable to those of hardware weak memory models \cite{sor11,sew10,alg14,col21}. Hence, we use an existing verification technique for reasoning about weak memory models, {\em reordering interference freedom (rif)\/} \cite{cou21,cou23}, to reason about effects of persistent memory. This technique
\begin{itemize}
    \item is applicable to all currently available processor memory models, and
    \item allows us to separately reason about
    \begin{itemize}
        \item the security of a program,
        \item the effects on security introduced by the processor's memory model, and
        \item the effects on security introduced by persistent memory.
    \end{itemize}
\end{itemize}
Not only does this simplify the overall reasoning, but enables reuse of reasoning about a program on a range of different processor architectures. It provides a notion of {\em relative security\/} which allows us to compare information flow in the presence and absence of particular architectural features (as in the work of Cheang et al. \cite{che19}). 

The paper makes two main contributions.

\begin{enumerate}
\item We provide an approach to information flow checking on persistent memory. To the best of our knowledge, this is the first such approach. 
\item  We demonstrate that the proof technique, reordering information freedom (rif), can be used to reason about persistent memory without having to develop specialised reasoning techniques.
\end{enumerate}

We begin in Section~\ref{sec:wpif} by providing an information flow logic for an unstructured programming language, i.e., one using gotos rather than higher level language constructs such as loops. We illustrate the use of the logic on a case study based on the Linux reader-writer mechanism seqlock \cite{bov05}. 
In Section~\ref{sec:rif}, we introduce rif and show how it can be used to reason about instruction reordering on x86 processors, and to additionally reason about effects due to inconsistent states possible with persistent memory. We demonstrate this reasoning using the case study of Section~\ref{sec:wpif}. A proof of the soundness of the approach is provided in Section~\ref{sec:sound}. In Section~\ref{sec:rel}, we provide an overview of related work before concluding in Section~\ref{sec:con}.

\section{Information flow logic}
\label{sec:wpif}

Our information flow logic is defined over an abstract syntax representing assembly language constructs. 
A procedure is defined as having an identifier and a set of basic blocks (including both a start block and one or more exit blocks). Each basic block comprises a (possibly empty) sequence of instructions followed by either a jump to another basic block in the procedure or, in the case of an exit block, a return. Jumps may be conditional allowing for two or more possible jump targets. The condition $b$ is in terms of registers and literal values only. 

\begin{sidebyside}
\[
procedure ::= id\!: block^+\\
block ::= id\!: instruction^*\, jump\]
\nextside
\[\hspace{-2mm}instruction ::= r := e | r := [x] | [x] := e\\
\hspace{-2mm}jump ::= \M \goto id | \ifc (b)\;  jump \elsec jump|\\
\return\O\]
\end{sidebyside}

There are three types of instructions: register updates $r := e$ where $r$ is a thread-local register name and $e$ is an expression in terms of registers and literals, loads $r := [x]$ where $x$ is the address of the data to be loaded into register $r$, and stores $[x] := e$. The address $x$ corresponds to a global, i.e., shared, program variable, or an offset from such a variable (representing an index into an array or other data structure). Each of these instructions is considered to be atomic in our abstract language.

Our logic builds on that of Winter et al.\ \cite{win21} which uses a combination of weakest precondition reasoning \cite{dij76,dij90} to check security proof obligations and rely/guarantee reasoning \cite{jon83,xu97} to take into account the effects of concurrently running threads. The memory is assumed to be partitioned in such a way that an attacker can only {\em directly\/}  access data labelled as $low$. Data labelled as $high$ is regarded as sensitive and it is the objective of the logic to detect any possibility of such data being directly or indirectly imparted to an attacker. 

We adapt that logic (developed for a simple while language) to our abstract assembly language syntax. Each memory location $x$ (corresponding to a program variable or part of a data structure) is given a security classification $\LL(x)$. The classification is either $high$ (meaning the location is not directly accessible by an attacker and can hold either $high$ or $low$ data) or $low$ (meaning the location is directly accessible by an attacker and hence should only hold $low$ data). The security level of the data itself is given by an auxiliary variable $\Gamma_x$ for a location $x$. Similarly, $\Gamma_r$ is used for the security level of the data held in a register $r$.

Given a Boolean lattice with values $high$ and $low$ such that $low \sqsubseteq high$ and $high \not\sqsubseteq low$, our logic checks that an instruction or jump carried out in a secure state, i.e., one where for all locations $x$, $\Gamma_x \sqsubseteq \LL(x)$, results in a secure state.\footnote{The approach can be extended to a general lattice as detailed in \cite{win21}.} As has been proved in Isabelle/HOL \cite{win21}, this amounts to verifying {\em noninterference\/} \cite{gog82}, i.e., that $low$ data is not influenced by $high$ data. The required check is made via proof obligations introduced by the weakest precondition rules (see below) for each instruction or jump.

Following Barnett and Leino \cite{bar05}, we apply weakest precondition reasoning to each basic block. The reasoning starts from a postcondition at each exit block, and conditions that need to hold to ensure noninterference are added by rules associated with each instruction and jump.

Given $id_{ok}$ denotes the weakest precondition of block $id$, the rules for jumps are as follows (where $Q$ is the poststate from which the rule is applied and $\wpif$ is our extension of standard weakest precondition reasoning). 

\[\wpif(\goto id, Q) ~=~ id_{ok}\\
\wpif(\ifc (b)\; j_1 \elsec j_2, Q) ~=~\M \Gamma_E(b) = low \land\\
(b \implies \wpif(j_1, Q)) \land (\neg b \implies \wpif(j_2, Q))\O\\
\wpif(\return, Q) ~=~ Q\]
where $\Gamma_E(b)$ evaluates to $high$ when any of the values that are used to express $b$ are $high$, and $low$ otherwise. The predicate $\Gamma_E(b) = low$ is a proof obligation ensuring that branching does not occur based on $high$ data: in concurrent programs, the value of $b$ can readily be deduced using timing attacks (even when the statement's branches do not change publicly accessible variables) \cite{mur18,smi19}. Consider, for example, the following procedure where $id_0$ is the start block and the basic block $id_1$ repeatedly performs the sequence of instructions $c$ until condition $b$ is false.  

\begin{sidebyside}[3]
\[id_0\!:\\
~~\M [x] := 0\\
\ifc (r=0)\; \goto id_1 \elsec \goto id_2\O\] 
\nextside
\[\t6 \M id_1\!:\\
~~\M c\\
\ifc (b)\; \goto id_1 \elsec \goto id_2\O\O\]
\nextside
\[\t9 \M id_2\!:\\
~~\M [x] := 1\\
\return\O\O\]
\end{sidebyside}

Assume that $r$ holds $high$ data. Even though the values stored to $[x]$ are independent of the value of $r$, a concurrent thread which simply loads $[x]$ will be able to deduce information about $r$ in some circumstances. For example, under a round robin scheduler with time slices less than the time it takes to execute $id_1$, loading the value~1 would indicate that $r \neq 0$. 

In circumstances where releasing {\em some\/} information about $high$ data is not problematic from a security perspective, the verification can include declassification of {\em what\/} expressions in terms of $high$ variables can be released, and {\em where\/} in the program this can occur \cite{sab09}. For example, when reasoning about the above procedure, we could declassify the expression $r = 0$ at line~2 of $id_0$ (and nowhere else). This expression is then treated as if it were $low$. Such an approach changes the property being proved to something slightly weaker, but ultimately more practical, than standard noninterference as detailed in \cite{ask07,smi22}.

For a register update $r:=e$, we have no additional proof obligation since registers are thread-local and hence no data is put into a location that is accessible by another thread. The rule differs from the standard weakest precondition rule for assignments by including an update to the auxiliary variable $\Gamma_r$ denoting the security level of the data in $r$.

\[\wpif(r:=e, Q) ~=~ Q[r, \Gamma_r \backslash e, \Gamma_E(e)]\]
where $Q[r, \Gamma_r \backslash e, \Gamma_E(e)]$ denotes predicate $Q$ with each occurrence of $r$ and $\Gamma_r$ replaced by $e$ and $\Gamma_E(e)$, respectively.

Load and store instructions add proof obligations that depend on globally accessible memory locations. These proof obligations can therefore be falsified by concurrently running threads. To allow for thread-local reasoning, we include an additional proof obligation that assures they are not falsified by any possible step, or sequence of steps, of threads in the environment. 

This potential environment behaviour is captured by a $rely$ predicate $\R$ \cite{jon83}. $\R$ is a two-state predicate over locations, i.e., global variables, describing the relation between the state before and after all possible environment steps. This relation is reflexive to capture that the environment may do nothing, and transitive to capture that the environment may perform several steps.

Each variable $v$ in the before state of $\R$ is denoted by $v'$ in the after state. Using $P'$ to refer to $P[\bar{v}\backslash \bar{v}']$ where $\bar{v}$ is the list of all variables in $P$ and $\bar{v}'$ their primed counterparts, a predicate $P$ is defined to be {\em stable\/}, i.e., non-falsifiable, under behaviour described by $\R$ as follows.

\[stable_\R(P) ~=~ \forall \bar{v}' \cdot P \land {\cal R}  \implies P'\]
A wellformedness condition of $\wpif$ requires that the specified precondition and postcondition of a procedure are stable. 

For a load $r:= [x]$, we require that $r$ and $\Gamma_r$ are updated, the latter to $\LL(x) \sqcap \Gamma_x$, i.e., the meet, or lowest value, of $\LL(x)$ and $\Gamma_x$. This ensures $\Gamma_r$ is no greater than $\LL(x)$ in the case where $\Gamma_x$'s value is not defined. This reflects that we are starting the execution of the load from a secure state. Additionally, we require that this proof obligation is stable.

\[\wpif(r := [x], Q) ~=~ Q[r, \Gamma_r\backslash [x], \LL(x) \sqcap \Gamma_x] \land stable_\R(Q[r, \Gamma_r\backslash [x], \LL(x) \sqcap \Gamma_x])\]

A store $[x] := e$ changes the data in the globally accessible location $x$. Hence, it is required that this change satisfies the rely predicates of all other threads in the environment. To support proving this, a thread has a $guarantee$ predicate $\G$ which stores must satisfy \cite{jon83}. Like $\R$, $\G$ is a two-state predicate over locations capturing a change of global state.

The rule for a store updates $[x]$ and 
$\Gamma_x$ and checks that the security value of $e$ is less than or equal to the security classification of $x$.
Additionally, it assures the thread's guarantee holds when the value of $[x]$ is updated to $e$, and that all proof obligations are stable.

\[\wpif([x] := e, Q) ~\M=~  Q[[x], \Gamma_x\backslash e, \Gamma_E(e)] \land \Gamma_E(e) \sqsubseteq \LL(x) \land \G[[x]'\backslash e][\bar{v}'\backslash \bar{v}]\land\\
\also
stable_\R(Q[[x], \Gamma_x\backslash e, \Gamma_E(e)] \land \Gamma_E(e) \sqsubseteq \LL(x) \land \G[[x]'\backslash e][\bar{v}'\backslash \bar{v}])\O\]
where $\bar{v}$ is the list of all variables in $\G$, i.e., $\G$ has all occurrences of $[x]'$ replaced by $e$ and all other primed variables $v'$ replaced by $v$ (capturing that the value of $v$ is unchanged). 

Finally, to reason over the sequence of instructions and jump in a basic block, we have the following rule (where $i_1$ and $i_2$ are sequences of instructions, $i_2$ optionally followed by a jump).

\[wpif(i_1 ~i_2, Q) ~=~ wpif(i_1, wpif(i_2, Q))\]

\subsection{Case study}
\label{sec:case}

To illustrate the use of the logic, we present a case study based on the Linux reader-writer mechanism seqlock \cite{bov05}. The case study provides a proof of concept that (i)~persistent memory can introduce new information flow problems, and (ii) that the proposed approach can find them. A deeper evaluation of the approach requires automated tool support which is left to future work. 

The seqlock mechanism allows reading of shared memory locations without the need for locking, thus supporting fast write access. We assume there is just one writer thread but many reader threads. When the writer thread wishes to write to the shared locations $x1$ and $x2$, it increments a counter $c$ whose initial value is even. It then proceeds to write to the locations, and finally increments $c$ again. The counter $c$ ensures the consistency of values read by other threads. The two increments of $c$ ensure that it is odd when a thread is writing to the locations, and even otherwise. Hence, when a thread wishes to read the locations, it waits in a loop until $c$ is even before reading them. Also, before returning it checks that the value of $c$ has not changed (i.e., another write has not begun). If it has changed, the values from $x1$ and $x2$ that have been read may not belong to the same write and hence the thread discards these values and starts over. 

To introduce a security element to our case study, we assume the location $x1$ may hold $high$ data, and that this is flagged via $x2$ which is 1 when the data at $x1$ is $high$ and 0 when it is $low$.

Assuming the values to be written to $x1$ and $x2$ are stored in registers $r1$ and $r2$, respectively, the $write$ procedure has just the following basic block.

\[wr_0\!:\\
~~\M r0 := [c]\\
[c] := r0 + 1\\
[x1] := r1\\
[x2] := r2\\
r0 := [c]\\
[c] := r0 + 1\\
\return\O\]

The $read$ procedure places the values at $x1$ and $x2$ in registers $r1$ and $r2$, respectively. Since registers are local to a thread, these registers are different to those of the writer thread and all other reader threads.

The procedure has four basic blocks. In the start block $rd_0$, the value at $c$ is read and, if it is even, the code jumps to basic block $rd_1$, otherwise it starts over. In $rd_1$, the value at $x2$ is read and if it is 1 the procedure jumps to basic block $rd_2$ where it clears the value in $r1$ and returns. The 1 in $r2$ indicates to the thread that called the procedure that the read has failed (due to the current value being $high$).\footnote{We assume there are additional {\em privileged\/} reader threads that could read the value at $x1$ in these cases.} Otherwise, the code jumps to basic block $rd_3$. This block reads the data at $x1$ then re-reads the value at $c$. If this latter value hasn't changed, then the values in $r1$ and $r2$ correspond to a single write and the procedure returns. In this case, the 0 in $r2$ indicates to the calling thread that the read was successful (and $r1$ contains the low data read from $x1$). If the value at $c$ has changed, the data cannot be returned and so the procedure starts over. 

\begin{sidebyside}
\[rd_0\!:\\
~~\M r0 := [c]\\
\ifc (r0 \mod 2 = 0)\; \goto rd_1 \elsec \goto rd_0\O\]
\nextside
\[\hspace{10mm} \M rd_1\!:\\
~~\M r2 := [x2]\\
\ifc (r2 = 1)\; \goto rd_2 \elsec \goto rd_3\O\O\]
\end{sidebyside}
\vspace*{-6mm}
\begin{sidebyside}
\[rd_2\!:\\
~~\M r1 := 0\\
\return\O\]
\nextside
\[\hspace{10mm}\M rd_3\!:\\
~~\M r1 := [x1]\\
r3 := [c]\\
\ifc (r3 = r0)\; \return \elsec \goto rd_0\O\O\]
\end{sidebyside}

To reason about information flow in this program, we let $\LL(c) = \LL(x2) = low$ as these locations do not hold sensitive data, and $\LL(x1) = high$ since it may hold $high$ information (when the value at $x2$ is 1).

The writer thread can rely on the values at $x1$, $x2$ and $c$ not being changed by other threads. It guarantees that the correspondence between the value at $x2$ and the security of data in $x1$ is maintained when the value at $c$ is even ($G\land G'$ below).  It also guarantees that it only ever increases the value at $c$, and does not change the values at $x1$ and $x2$ when the value at $c$ is even. 

\[\R_{wr} = ([x_1]' = [x_1] \land [x_2]'=[x_2] \land [c]'=[c])\\
\G_{wr} = G \land G' \land [c]' \geq [c] \land (c\mod 2 = 0 \implies [x1]'=[x1]\land[x2]'=[x2])\\
\mbox{where~} G = ([c]\mod 2=0 \implies ([x2]=1 \iff \Gamma_{x1} = high))\]

\begin{figure}[t]
\[\!\!\!\!wr_0\!:\\
\{\Gamma_{r2} = low \land [c]\mod 2 = 0 \land (r2=1 \iff \Gamma_{r1} = high) \land ([x2]=1\iff \Gamma_{x1}=high)\}\\
    r0 := [c]\\
\{\M \Gamma_{r0} = low \land \Gamma_{r2} = low \land r0\mod 2=0 \land (r2=1 \iff \Gamma_{r1} = high)\land r0+1 \geq [c]\land\\
([c] \mod 2 = 0 \implies ([x2]=1\iff \Gamma_{x1}=high))\O\}\\
    [c] := r0 + 1\\
\{\Gamma_{r2} = low \land ([c]+1)\mod 2=0 \land (r2=1 \iff \Gamma_{r1} = high)\}\\
    [x1] := r1\\
\{\Gamma_{r2} = low\land ([c]+1)\mod 2=0 \land (r2=1 \iff \Gamma_{x1} = high)\}\\
    [x2] := r2\\
\{([c]+1)\mod 2=0 \land ([x2]=1 \iff \Gamma_{x1} = high)\}\\
    r0 := [c]\\
\{\Gamma_{r0}=low \land (r0+1)\mod 2=0 \land ([x2]=1 \iff \Gamma_{x1} = high) \land r0+1 \geq [c]\}\\
    [c] := r0 + 1\\
\{[c] \mod 2 =0\land ([x2]=1 \iff \Gamma_{x1} = high)\}\]
\vspace{-4mm}
\caption{Calculation of weakest precondition of $write$ procedure.}
\label{fig:write}
\end{figure}

\begin{figure}[h!]
\[\!\!\!\!rd_0\!:\\
\{[c]\mod 2=0 \implies ([x2]\neq 1 \implies \Gamma_{x1}=low)\}\\
    r0 := [c]\\
\{\M \Gamma_{r0} = low \land \\
(r0\mod 2 = 0\implies\\
\t2 ([x2]\neq 1 \implies\M r0 \leq [c]\land ([c]=r0 \implies \Gamma_{x1} = low) \land\\
([c]\neq r0 \implies ([c]\mod 2= 0 \implies ([x2]\neq 1 \implies \Gamma_{x1}=low))))\O\land\\
(r0\mod 2\neq 0 \implies ([c]\mod 2=0 \implies ([x2]\neq 1 \implies \Gamma_{x1}=low)))\O\}\\
    \ifc (r0 \mod 2 = 0)\; \goto rd_1 \elsec \goto rd_0\\
\{true\}\]
\vspace*{-10mm}
\[\!\!\!\!rd_1\!:\\
\{([x2]\neq 1 \implies\M r0 \mod 2 = 0 \land r0 \leq [c]\land \Gamma_{r0} = low\land ([c]=r0 \implies \Gamma_{x1} = low) \land\\
([c]\neq r0 \implies ([c]\mod 2= 0 \implies ([x2]\neq 1 \implies \Gamma_{x1}=low)))\O\}\\
    r2:=[x2]\\
\{\M \Gamma_{r2} = low \land\\
(r2\neq 1 \implies\M r0 \mod 2 = 0 \land r0 \leq [c]\land \Gamma_{r0} = low\land ([c]=r0 \implies \Gamma_{x1} = low)\land\\
([c]\neq r0 \implies ([c]\mod 2= 0 \implies ([x2]\neq 1 \implies \Gamma_{x1}=low)))\O\O\}\\
    \ifc (r2=1)\; \goto rd_2 \elsec \goto rd_3\\
\{true\}\]
\vspace*{-10mm}
\[\!\!\!\!rd_2\!:\\
\{true\}\\
r1 := 0\\
\{\Gamma_{r1}=low\}\\
\return\\
\{\Gamma_{r1}=low\}\]
\vspace*{-10mm}
\[\!\!\!\!rd_3\!:\\
\{\M r0 \mod 2 = 0 \land r0 \leq [c]\land \Gamma_{r0} = low\land
([c]=r0 \implies \Gamma_{x1} = low) \land\\
([c]\neq r0 \implies ([c]\mod 2= 0 \implies ([x2]\neq 1 \implies \Gamma_{x1}=low)))\O\}\\
    r1 := [x1]\\
\{\M r0 \leq [c]\land \Gamma_{r0} = low\land
([c]=r0 \implies \Gamma_{r1} = low) \land\\
([c]\neq r0 \implies ([c]\mod 2= 0 \implies ([x2]\neq 1 \implies \Gamma_{x1}=low)))\O\}\\
    r3 := [c]\\
\{\M \Gamma_{r3}=low \land \Gamma_{r0} = low\land
(r3=r0 \implies \Gamma_{r1} = low) \land\\
(r3\neq r0 \implies ([c]\mod 2= 0 \implies ([x2]\neq 1 \implies \Gamma_{x1}=low)))\O\}\\
    \ifc (r3=r0)\; \return \elsec \goto rd_0\\
\{\Gamma_{r1}=low\}\]
\vspace{-4mm}
\caption{Calculation of weakest precondition of $read$ procedure.}
\label{fig:read}
\end{figure}

The application of the information flow logic to the $write$ procedure is given in Figure~\ref{fig:write} (where predicates at each step have been simplified). Starting with a postcondition that the value at $c$ is even and, hence, that $x2=1 \iff \Gamma_{x1}=high$, the proof applies the $\wpif$ rules to derive the weakest precondition to ensure noninterference and the thread's guarantee. Note that the stability conditions are trivially true at each step due to the environment threads not changing the values at $c$, $x1$ or $x2$. 

Since none of the predicates in the proof evaluates to false, the procedure is secure when started in a state satisfying the calculated weakest precondition; that is, in a state where $r2$ holds $low$ data, the value at $c$ is even, $r1$ holds $high$ data if, and only if, $r2=1$, and $[x1]$ holds $high$ data if, and only if, $[x2] = 1$. The latter condition ensures that the guarantee is met by the first store.

The rely and guarantee predicates for the $read$ procedure are the inverse of those for $write$, i.e., $\R_{rd} = \G_{wr}$ and $\G_{rd}=\R_{wr}$. This ensures compatibility with the writer thread (i.e., it ensures that the writer thread ensures what the reader threads rely on, and vice versa). Also, since the $read$ procedure only changes (thread-local) registers, each thread calling it is trivially compatible with other reader threads. 

The information flow logic is applied to the $read$ procedure in Figure~\ref{fig:read} (where again predicates have been simplified at each step). Starting with a postconditon that $r1$ holds $low$ information (and hence there is no information leak), the weakest precondition of the start block $rd_0$ is calculated to be the predicate $[c]\mod 2= 0 \implies ([x2]\neq1\implies \Gamma_{x1}=low)$ which is implied by the predicate $G$ maintained by the writer thread. Note that due to cycles between the basic blocks, a process akin to deriving loop invariants needs to be used to calculate the basic block's preconditions. 

At each step, the predicates referring to locations which can be modified by the writer are stable as follows. The predicate $[c]\mod 2= 0 \implies ([x2]\neq 1 \implies \Gamma_{x1}=low)$ is implied by predicate $G$ which is maintained by the writer thread. Similarly, each predicate with $[c]\mod 2= 0 \implies ([x2]\neq 1 \implies \Gamma_{x1}=low)$ as the consequent of an implication is stable. 

The stability of predicate $[c]=r0 \implies \Gamma_{r1} = low$ reduces to $r0 \leq [c]$: since the writer guarantees $[c]' \geq [c]$, the antecedent of the predicate can never go from $false$ to $true$ when $r0 \leq [c]$ and hence the predicate cannot be falsified. Similarly, the stability of the predicate $[c]=r0 \implies \Gamma_{x1} = low$ reduces to $r0 \leq [c] \land r0 \mod 2 = 0$: again $r0 \leq [c]$ ensures that the antecedent does not go from $false$ to $true$, and $r0\mod 2 = 0$ ensures $[c]\mod 2 = 0$ when the antecedent is true and hence that the writer guarantees $x1$ is not changed. 

Again since none of the predicates in the proof evaluate to $false$, the procedure is secure when started from a state satisfying its weakest precondition. The weakest precondition of $read$ requires that if the value at $c$ is even then if $[x2] \neq 1$, $x1$ contains $low$ data (since it is this data which will be returned in $r1$).
 
\section{Reordering interference freedom (rif)}
\label{sec:rif}

The analysis of the seqlock case study in Section~\ref{sec:case}, did not take into account the memory model of the processor on which the code is running, nor the effect of using persistent memory in the case of an unplanned power outage. 

Recently, custom approaches have been developed to enable such reasoning \cite{raa20b,bil22a}. Our approach, however, is to use the rely/guarantee-based analysis already presented and to augment it with side-conditions to identify new vulnerabilities introduced by given architectural features. 

We use the notion of {\em reordering interference freedom\/}, or $rif$ for short, introduced by Coughlin et al.\ \cite{cou21,cou23} to reason about the effects of processor weak memory models. As we show in this section, it can also be used to reason about the effects of persistent memory. 

\subsection{Definition of rif}
\label{sec:defn}

In their work on operational semantics for hardware weak memory models, Colvin and Smith \cite{col18a,col18b} show that the effects of all {\em multicopy atomic memory models\/}, i.e., those in which all processor cores see writes to memory at the same time, can be captured in terms of an instruction reordering relation $R$. This is validated by comparing simulations of the semantics against actual hardware using the widely accepted range of litmus tests of Alglave et al.\ \cite{alg14}.

$R$ captures the reordering of memory accesses, i.e., stores and loads, allowed by a given memory model to improve efficiency. Such reordering is required to maintain a thread's in-order semantics and hence has certain restrictions (detailed in \cite{col18a,col18b}). In general, two instructions cannot reorder when there is a dependency between them. For example, the program $r:=0~\, [x]:=r$ cannot be reordered to $[x]:=r~\, r:= 0$ due to the store depending on the update to $r$. Importantly, reordered stores must write the values they would have written if executed in order. This can be assured by rewriting register updates/loads in dynamic single assignment (DSA) form, i.e., where each register variable is assigned only once in an execution \cite{van05}. For example, $r:= 0~\, [x]:=r~\, r:= 1$ would be rewritten as $r_0:=0~\, [x]:=r_0~\, r_1:=1$ ensuring that location $x$ is updated with value 0 even if the final register update is reordered before it. We assume $R$ does not allow non-memory instructions, i.e., register updates and jumps, to reorder with each other; such reorderings, while possible on a pipelined architecture, are inconsequential for the ordering of memory accesses.

In the reordering semantics, the rule for sequential composition is split into two cases: standard in-order execution, and reordered execution according to any reordering of instructions allowed by $R$.

\begin{sidebyside}[4]
\nextside
\[~\also\alpha~ c ~\stackrel{\alpha}{\longrightarrow}~ c\]
\nextside
\begin{infrule}
c ~\stackrel{\beta}{\longrightarrow}~c' \t3 (\alpha, \beta) \mem R
\derive 
\t2~~~\;\alpha~ c ~\stackrel{\beta_{\lseq \alpha\rseq}}{\relbar\joinrel\longrightarrow}~\alpha~ c'
\end{infrule}
\nextside
\end{sidebyside}

The rule for in-order execution (left-hand rule above) allows the next instruction $\alpha$ in a program $\alpha~ c$ to execute leaving the program $c$. The reordered execution (right-hand rule above) allows an instruction $\beta_{\lseq \alpha\rseq}$, which is not the next instruction, to execute transforming program $\alpha~ c$ to $\alpha~ c'$. This is possible whenever $c$ can execute $\beta$ to become $c'$ and $\alpha$ and $\beta$ are reorderable according to $R$. In this rule, $\beta_{\lseq \alpha\rseq}$ denotes a modification of $\beta$ which takes into account that it has been reordered with $\alpha$. For example, the code $[x] := 1~~ r := [x]$ can be reordered to $r := 1~~ [x] := 1$ on the Intel/AMD memory model x86 \cite{sew10}, where the load $r:= [x]$ has been modified to $r := 1$. Such modification is referred to as {\em forwarding\/} (or {\em bypassing\/}) and occurs to maintain the sequential semantics of the program \cite{sor11}. Formally, $\beta_{\lseq \alpha\rseq}$ is defined as follows.

\[\beta_{\lseq\alpha\rseq} =
\begin{cases}
      r:=e, & \text{if}\ \alpha=([x]:=e) \text{~and~} \beta=(r:=[x]) \text{~for some $x$, $r$ and $e$}\\
      \beta, & \text{otherwise}
\end{cases}\]


Note that the reordering rule can be applied successively to reorder, for example, $\alpha~\beta~\gamma~c$ to $\gamma~\alpha~\beta~c$ when $(\beta,\gamma)\mem R$ and $(\alpha, \gamma) \mem R$ (and no forwarding is required): first $\beta~\gamma~c$ is reordered to $\gamma~\beta~c$ then $\alpha~\gamma~\beta~c$ is reordered to $\gamma~\alpha~\beta~c$.

The semantics of other language constructs are not affected by $R$ and hence have the standard semantics as shown for a given procedure in Figure~\ref{fig:sem}. Note~that the $\return$ statement terminates the procedure and, as such, has no behaviour.

Based on the above semantics, the idea of rif is simple. Essentially the approach checks, for every pair of {\em reorderable instructions\/} of a program, i.e., instructions that can occur in the reverse order to the order they appear in the program, that executing the instructions in the reverse order does not introduce new behaviour. This check is defined for a verification approach based on rely/guarantee reasoning and includes ensuring that

\begin{enumerate}[(i)]
\item interference from other threads is taken into account via the program's rely condition,
\item the program's guarantee is maintained under the reordering, and 
\item other instructions reordering before or after the reorderable pair are taken into account by performing the check for the execution of the instructions with respect to an arbitrary pre- and post-state.
\end{enumerate}

\begin{figure}[t!]
\begin{sidebyside}[4]
\begin{infrule}
\hspace{-3mm}id\!: c ~\mem~ blocks\hspace{-3mm}
\derive
\hspace{-3mm}\goto id \stackrel{[true]}{\relbar\joinrel\relbar\joinrel\longrightarrow} c\hspace{-3mm}
\end{infrule}
\nextside
\begin{infrule}
\t3 \eval(\sigma,b)
\derive
\hspace{-3mm}\ifc(b)\; j_1 \elsec j2 \stackrel{[b]}{\longrightarrow} j1\hspace{-3mm}
\end{infrule}
\nextside
\begin{infrule}
\t1~~ \neg \eval(\sigma, b)
\derive
\hspace{-3mm}\ifc(b)\; j_1 \elsec j2 \stackrel{[\neg b]}{\relbar\joinrel\longrightarrow} j2\hspace{-3mm}
\end{infrule}
\nextside
\vspace*{-1.5mm}
\begin{infrule}
\hspace{-3mm}~~c ~\stackrel{[b]}{\longrightarrow}~ c'\hspace{-3mm}
\derive
\hspace{-3mm}\sigma, c \longrightarrow \sigma, c'\hspace{-3mm}
\end{infrule}
\end{sidebyside}
\begin{sidebyside}[3]
\begin{infrule}
\also
\t1~~~~\; c ~\stackrel{r:=e}{\relbar\joinrel\relbar\joinrel\longrightarrow}~ c'
\derive
\hspace{-3mm}\sigma, c \longrightarrow \sigma[r\mapsto {\sf eval}(\sigma,e)], c'\hspace{-3mm}
\end{infrule}
\nextside
\begin{infrule}
\t2 c ~\stackrel{r:=[x]}{\relbar\joinrel\relbar\joinrel\longrightarrow}~ c'
\derive
\hspace{-3mm}\sigma, c \longrightarrow \sigma[r\mapsto \sigma(x)], c'\hspace{-3mm}
\end{infrule}
\nextside
\begin{infrule}
\t1~~~~ c ~\stackrel{[x]:=e}{\relbar\joinrel\relbar\joinrel\longrightarrow}~ c'
\derive
\hspace{-3mm}\sigma, c \longrightarrow \sigma[x\mapsto {\sf eval}(\sigma,e)], c'\hspace{-3mm}
\also
\end{infrule}
\end{sidebyside}
For a given procedure, $blocks$ is the set of basic blocks of the procedure, and $\sigma$ is the state comprising values at global memory locations and in registers local to the thread calling the procedure ($\sigma[l \mapsto v]$ updates location/register $l$ with value $v$ and ${\sf eval}(\sigma,e)$ is the evaluation of expression $e$ in $\sigma$).
\vspace*{4mm}
\caption{Operational semantics of a procedure.}
\label{fig:sem}
\end{figure}

For the logic in Section~\ref{sec:wpif}, we check the following for each pair of reorderable instructions $(\alpha, \beta)$ for a program with rely $\R$ and guarantee $\G$.

\[\t1 rif(\alpha,\beta, \R, \G)  ~\sdef~  \all Q\cdot \wpif(\alpha~ \beta, Q) \implies \wpif(\beta_{\lseq\alpha\rseq}~ \alpha, Q) & (1)\]
That is, for every state from which the execution of $\alpha~ \beta$ will reach an arbitrary state $Q$, the execution of $\beta_{\lseq\alpha\rseq}~\alpha$ will also reach $Q$. The use of $\wpif$ ensures that the program's rely and guarantee are taken into account (conditions (i) and (ii) above), and the universal quantification over $Q$ over-approximates the post-states the executions result in (ensuring condition (iii)).

The over-approximation ensures the soundness of rif, i.e., that it captures all behaviours possible under the given reordering, as has been proved in Isabelle/HOL~\cite{cou21}. While false positives are possible, they are not common (most allowed reorderings are benign in any context) and are readily dealt with using strategies detailed in \cite{cou21}.

The over-approximation also dramatically reduces the complexity of reasoning in the presence of reordering. For a thread with $n$ instructions, the worst case is that every instruction can reorder giving us $n(n-1)/2$ reorderable pairs (significantly less than the $n!$ execution traces that such reordering would introduce). Note also that this worst case is extremely unlikely. Weak memory models preserve sequential semantics and hence, apart from cases of forwarding, instructions which refer to the same variable are not reorderable.

The rif approach is also readily automated. Pairs of reorderable instructions of a program can be found via a dataflow analysis, similar to dependence analysis commonly used in compiler optimisation \cite{kil73}. This approach has been implemented for both a simple while language and an abstraction of ARMv8 assembly code \cite{cou21}.

The operational semantics and definition of rif above is for multicopy atomic processors. These include widely used x86 and ARMv8 processors as well as those based on the open-source RISC-V architecture.
An extension of the semantics has been developed for non-multicopy atomic processors such as IBM POWER and ARMv7 \cite{col18a,col18b}. The interested reader is referred to \cite{cou23} for a corresponding definition of rif which works for these processors. 

\subsection{rif on the x86 weak memory model}
\label{sec:tso}

To illustrate the use of rif, we apply it to analyse information flow security of our seqlock case study when running on an x86 processor. On x86 (also referred to as Total Store Order (TSO)), stores occur in program order but may be delayed with respect to loads, register updates and jumps. In the case of loads to the same memory location as the store, forwarding occurs (as discussed in Section~\ref{sec:defn}). These rules are summarised in Table~\ref{tab:x86} where X indicates no reordering is possible, F indicates reordering is possible and that forwarding applies in the case that the store and load refer to the same location, and $\checkmark$ indicates reordering is possible. Instruction~1 is the earlier instruction and Instruction~2 the later. For greater completeness, we have included atomic read-modify-write instructions such as compare-and-swap and fetch-and-add ($\RMW$) and a general fence instruction ($\mfence$) used to prevent unwanted reordering.\footnote{x86 also supports a store fence (sfence) and  a load fence (lfence) which we omit here.} The operational semantics and $\wpif$ rules for these constructs are provided below. We model a general $\RMW$ instruction as $\RMW(b,x,e_1,e_2)$ which updates the location $x$ to $e_2$ when $b$ holds. For a fetch-and-add with parameters $x$ and $e_1$, $b$ is true and $e_2 = \sigma([x])+e_1$. For a compare-and-swap with parameters $x$, $e_1$ and $e_2$, $b$ is $\sigma([x])=e_1$.

\begin{table}[t]
    \begin{tabular}{ccccccc}\toprule
           & \multicolumn{6}{c}{~~~~~~~Instruction 2} \\[.5ex] 
          && ~~load~~  & ~~store~~ & ~~RMW~~ & ~~mfence~~ & ~~other~~ \\[.5ex]
          & load		& X & X & X & X & X\\
          & store 	& F & X & X & X & \checkmark \\
          & RMW 	& X & X & X & X & X \\
            \rot{\rlap{\hspace{-4mm}Instruction 1}}
          & mfence	& X & X & X & X & X \\
          & other  & X & X & X & X & X\\ \bottomrule
    \end{tabular}
\caption{Instruction reordering on x86 (from \cite{sor11}, extended with the category ``other'' for register updates and jumps).}
\label{tab:x86}
\end{table}

\begin{sidebyside}
\begin{infrule}
c \stackrel{\mfence}{\relbar\joinrel\relbar\joinrel\relbar\joinrel\longrightarrow} c'
\derive
\sigma, c \longrightarrow \sigma, c'
\end{infrule}
\nextside
\begin{infrule}
\t1~ {\sf eval}(\sigma, b) \t2 c \stackrel{\RMW(b,x,e_1,e_2)}{\relbar\joinrel\relbar\joinrel\relbar\joinrel\relbar\joinrel\relbar\joinrel\relbar\joinrel\relbar\joinrel\relbar\joinrel\longrightarrow} c'
\derive
\sigma, c \longrightarrow \sigma[x\mapsto {\sf eval}(\sigma, e_2), r \mapsto \sigma(x)], c'
\end{infrule}
\end{sidebyside}
\begin{sidebyside}
\begin{infrule}
\neg{\sf eval}(\sigma, b) \t2 c \stackrel{\RMW(b,x,e_1,e_2)}{\relbar\joinrel\relbar\joinrel\relbar\joinrel\relbar\joinrel\relbar\joinrel\relbar\joinrel\relbar\joinrel\relbar\joinrel\longrightarrow} c'
\derive
\t1~~~~ \sigma, c \longrightarrow \sigma[r\mapsto \sigma(x)], c'
\end{infrule}
\nextside
\vspace*{1.5mm}
\hspace*{6mm} \parbox{55mm}{\small where $r$ is the x86 register into which $[x]$ is loaded during a read-write-modify instruction.}
\end{sidebyside}

\[\wpif(\mfence, Q) = Q\\
\Also
\wpif(\RMW(b,x,e_1,e_2), Q) = \Gamma_E(b) = low \land (b \implies \wpif([x]:=e_2,Q)) \land (\neg b \implies Q)\]
where $r$ is not free in $Q$ nor $e_2$ (and hence loading $[x]$ to it has no effect on the calculated weakest precondition). Note that neither the operational semantics nor $\wpif$ rules take into account the ordering constraints introduced by these instructions, as these will be handled by rif.

Table~\ref{tab:x86} represents the reordering relation $R$ where for each $(\alpha,\beta)\mem R$, $\alpha$ is a store $[x]:=e_1$ and $\beta$ is either a 
load $r:=[y]$, register update $r:=e_2$ or jump. 

%

Recalling the $write$ procedure of seqlock from Section~\ref{sec:case}, three pairs of reorderable instructions can be found. The second load $r0 := [c]$ can be reordered with each of the proceedings stores. 


We apply the rif check (1) to the first pair of instructions $[x2]:= r2$ and $r0 := [c]$ below. Noting that $write$'s rely condition ensures values at $c$, $x1$ and $x2$ do not change, we have the following (where $r$, $r_0$ and $r_1$ are the DSA variables corresponding to the initial value, first and second assignments to a register $r$, respectively).

\[~~~\wpif([x2] := r2~~ r0_1:=[c], Q)~ \\
= \wpif([x2]:=r2, \wpif(r0_1:=[c],Q)) & (sequential composition)\\
= \wpif([x2]:=r2,Q[r0_1,\Gamma_{r0_1}\backslash [c],\Gamma_{c}]) & (load)\\
= \M Q[r0_1,\Gamma_{r0_1}\backslash [c],\Gamma_{c}]\,[[x2], \Gamma_{x2}\backslash r2, \Gamma_{r2}] \land\\
([c]\mod 2 = 0 \implies [x2] = r2 \land ([x2]=1\iff\Gamma_{x1}=high)) & (store)\O\\
\also
\zbreak
~~~\wpif(r0_1:=[c]~~ [x2] := r2, Q)\\
= \wpif(r0_1:=[c], \wpif([x2] := r2,Q)) & (sequential composition)\\
=  \wpif(\M r0_1:=[c],Q[[x2],\Gamma_{x2}\backslash r2,\Gamma_{r2}]\land\\
([c]\mod 2 = 0 \implies [x2] = r2 \land ([x2]=1\iff\Gamma_{x1}=high)))\O & (store)\\
= \M Q[[x2], \Gamma_{x2}\backslash r2, \Gamma_{r2}]]\,[r0_1,\Gamma_{r0_1}\backslash [c],\Gamma_{c}]\land\\
([c]\mod 2 = 0 \implies [x2] = r2 \land ([x2]=1\iff\Gamma_{x1}=high)) & (load)\O\]

Since neither updated variable depends on the other, the order of substitution does not matter. Hence, $\wpif([x2] := r2~\, r0_1:=[c], Q) = \wpif(r0_1:=[c]~\, [x2] := r2, Q)$ and (1) holds. From this we deduce that no new behaviour is introduced by the reordering and the results of the analysis of Section~\ref{sec:case}, that seqlock is secure, still holds. The rif check for the reorderable pair $[x1]:=r1$ and $r0:=[c]$ similarly holds. 

Forwarding is required for the pair $[c]:=r0+1$ and $r0:=[c]$ since the store and load refer to the same location $c$. The rif check is as follows.

\[~~~\wpif([c] := r0_0+1~~ r0_1:=[c], Q)~ \\
= \wpif([c]:=r0_0+1, \wpif(r0_1:=[c],Q)) & (sequential composition)\\
= \wpif([c]:=r0_0+1,Q[r0_1,\Gamma_{r0_1}\backslash [c],\Gamma_{c}]) & (load)\\
= \M Q[r0_1,\Gamma_{r0_1}\backslash [c],\Gamma_{c}]\,[[c], \Gamma_{c}\backslash r0_0+1, \Gamma_{r0_0}] \land r0_0+1 \geq [c] \land \\
([c]\mod 2 = 0 \implies ([x2]=1\iff\Gamma_{x1}=high)) \land\\
((r0_0+1)\mod 2 = 0 \implies ([x2]=1\iff\Gamma_{x1}=high))& (store)\O\\
\also
\zbreak
~~~\wpif(r0_1:=r0_0+1~~ [c] := r0_0+1, Q) & (forwarding applied)\\
= \wpif(r0_1:=r0_0+1, \wpif([c] := r0_0+1,Q)) & (sequential composition)\\
=  \wpif(\M r0_1:=r0_0+1,Q[[c],\Gamma_{c}\backslash r0_0+1,\Gamma_{r0_0}]\land r0_0+1 \geq [c] \land\\
([c]\mod 2 = 0 \implies ([x2]=1\iff\Gamma_{x1}=high))\land\\
((r0_0 + 1) \mod 2 = 0 \implies ([x2]=1\iff\Gamma_{x1}=high)))\O & (store)\\
= \M Q[[c], \Gamma_{c}\backslash r0_0+1, \Gamma_{r0_0}]]\,[r0_1,\Gamma_{r0_1}\backslash r0_0+1,\Gamma_{r0_0}]\land r0_0+1 \geq [c] \land\\
([c]\mod 2 = 0 \implies ([x2]=1\iff\Gamma_{x1}=high))\land\\
((r0_0+1)\mod 2 = 0 \implies ([x2]=1\iff\Gamma_{x1}=high)))& (load)\O\]

Again the order of substitutions does not matter. In this case, for both orders $[c]$ and $r0_1$ become $r0_0+1$, and $\Gamma_c$ and $\Gamma_{r0_1}$ become $\Gamma_{r0_0}$. Hence, no new behaviour is introduced. Since there are no further reorderable pairs in $write$ and no reorderable pairs in $read$ (since it has no stores), it follows that seqlock is secure when executed on x86. 

\subsection{rif on persistent memory}
\label{sec:nvm}

A key problem with programming for persistent memory is ensuring a consistent state after a power outage. The problem arises due to the way stores are propagated from volatile to persistent memory. This is based on cache lines in the sense that stores on the same cache line are {\em persisted\/}, i.e., committed to persistent memory, together. More importantly, stores on different cache lines are committed at different times (determind by the hardware) which may not correspond to the order in which these stores occur in the program.

For seqlock, when a power outage occurs, we would like to restart in a state where reader threads can immediately begin calling the $read$ procedure. This will be true if the predicate $G=([c]\mod 2=0 \implies ([x2]=1\iff \Gamma_{x1}=high)$ holds, since $G$ implies the weakest precondition of $read$. It is not possible to also guarantee the weakest precondition of $write$ as this requires $[c]$ to be even, for example (and it is possible for a crash to occur when the latest persisted value at $c$ is odd). Hence, a start-up/recovery routine would need to be run before $write$ is called.

Our earlier analysis shows that, under normal execution, $G$ always holds. However, imagine 
that stores to $x1$, $x2$ and $c$ are placed on different cache lines. Given a write of a $high$ value to $x1$ after a previous write of a low value, $[x1]$ and the final (even) value at $c$ may be persisted, but not yet the value at $x2$ which indicates that the value in $x1$ is high. If at this time the computer crashes and needs to be restarted, the system will be in a state where $[c]$ is even and $[x2]$ is 0 and the value at $x1$ is high. That is, $G$ does not hold and a reader thread will be able to access $[x1]$ despite it being $high$. 

To help programmers avoid such problems, Intel-x86 persistent memory provides instructions $\flush$ and $\flush_{opt}$ for forcing cache lines to be committed \cite{raa20a}.\footnote{There is also an instruction $\wb$ (write back) which we do not consider: semantically, it behaves like $\flush_{opt}$ but performs better in some circumstances \cite{raa20a}.} 
The instruction $\flush ~x$ persists all pending stores to locations on the same cache line as $x$. The instruction $\flush_{opt} ~x$ also persists all pending stores to locations on the same cache line as $x$, but optimises performance by not necessarily executing at the point in the program where it occurs. It must occur after all earlier stores to locations on $x$'s cache line, but may occur after later stores and before earlier stores to locations on different cache lines.

These rules, as well as others related to the order of $\flush$ and $\flush_{opt}$ instructions are summarised in Table~\ref{tab:pm} where X indicates no reordering is possible, F indicates reordering is possible, however forwarding may need to be taken into account, CL indicates reordering is possible when the instructions refer to different cache lines, and $\checkmark$ denotes reordering is always possible.\footnote{The table corresponds to Intel's intended behaviour as opposed to the behaviour described in the under-specified manual text. See \cite{raa20a} for details.}

\begin{table}[t]
    \begin{tabular}{ccccccccc}\toprule
         & \multicolumn{8}{c}{~~~~~~~Instruction 2} \\[.5ex] 
                   && load & ~~store~~ & ~~RMW~~ & ~~mfence~~ & ~~\flush~~ & ~~\flush$_{opt}$~~ & ~~other~~\\[.5ex]
          & load		& X & X & X & X & X & X & X\\
          & store 	& F & X & X & X & X & CL & $\checkmark$\\
          & RMW 	& X & X & X & X & X & X & X\\
          & mfence	& X & X & X & X & X & X & X\\ 
                      \rot{\rlap{\hspace{-4mm}Instruction 1}}
          & flush	& $\checkmark$ & X & X & X & X & CL & \checkmark\\
          & flush$_{opt}$	& $\checkmark$ & $\checkmark$ & X & X & CL & $\checkmark$ & $\checkmark$\\
          & other & X &X & X & X & X & X & X\\ \bottomrule
    \end{tabular}
\caption{Instruction reordering on Intel-x86 persistent memory (from \cite{raa20a}, extended with the category ``other'' for register updates and jumps).}
\label{tab:pm}
\end{table}

The semantics for $\flush$ and $\flush_{opt}$ are similar to the $\mfence$ semantics; again we will use rif to handle ordering constraints.

\begin{sidebyside}[4]
\begin{infrule}
c \stackrel{\flush~x}{\relbar\joinrel\relbar\joinrel\relbar\joinrel\longrightarrow} c'
\derive
\sigma, c \longrightarrow \sigma, c'
\end{infrule}
\nextside
\begin{infrule}
c \stackrel{\flush_{opt}~x}{\relbar\joinrel\relbar\joinrel\relbar\joinrel\relbar\joinrel\relbar\joinrel\longrightarrow} c'
\derive
\t1 \sigma, c \longrightarrow \sigma, c'
\end{infrule}\nextside
\end{sidebyside}
\[\wpif(\flush~x, Q) = \wpif(\flush_{opt}~x, Q) = Q\]

To reason about behaviour over power outages, we extend the reordering relation captured in Table~\ref{tab:pm} to also allow stores to different locations to reorder. This is based on the observation that the state after a crash may have only the second of two successive stores in a program persisted. This is equivalent to a state reachable in the program when the stores can be reordered. Hence, to check that a program $c$ is secure in the presence of possible crashes, we need to prove it secure when any combination of stores to different locations can be reordered. That is, we need to check $rif(\alpha,\beta,\R,\G)$ (see (1)) for each pair of stores $(\alpha,\beta)$ to different locations in $c$. 

If the rif check succeeds, we have the following.

\[\all Q\cdot wpif(\alpha~\beta, Q) \implies \wpif(\beta~\alpha,Q)\]
That is, for all states from which the execution of $\alpha~\beta$ is secure (i.e., all instruction proof obligations hold, the guarantee is maintained and all states are stable), the execution of $\beta~\alpha$ is also secure. It follows that the execution of $\beta$ alone from such a state is secure (since during the execution of $\beta~\alpha$ concurrent threads will be able to observe the execution of $\beta$ and have access to the state immediately after it executes). Hence, if a power outage occurs when only $\beta$ is persisted, the state will be secure. 
Table~\ref{tab:out} is identical to Table~\ref{tab:pm} except for the cell for the ordering of a pair of stores, where L indicates reordering is possible when the instructions refer to different locations. 

\begin{table}[b]
    \begin{tabular}{ccccccccc}\toprule
       & \multicolumn{8}{c}{~~~~~~~Instruction 2} \\[.5ex] 
          && load & ~~store~~ & ~~RMW~~ & ~~mfence~~ & ~~\flush~~ & ~~\flush$_{opt}$~~ & ~~other~~\\[.5ex]
          & load	& X & X & X & X & X & X & X\\
          & store 	& F & L & X & X & X & CL & $\checkmark$\\
          & RMW 	& X & X & X & X & X & X& X\\
          & mfence	& X & X & X & X & X & X& X\\ 
                      \rot{\rlap{\hspace*{-4mm}Instruction 1}}
          & \flush	& $\checkmark$ & X & X & X & X & CL & $\checkmark$\\
          & \flush$_{opt}$	& $\checkmark$ & $\checkmark$ & X & X & CL & $\checkmark$ & $\checkmark$\\ 
          & other & X &X & X & X & X & X & X\\ \bottomrule
    \end{tabular}
\caption{Instruction reordering for reasoning with rif on Intel-x86 persistent memory.}
\label{tab:out}
\end{table}

The proposed approach is proved sound with respect to an existing approach of Raad et al.\ \cite{raa20a} in Section~\ref{sec:sound}. Here we illustrate it on the seqlock case study. 
For the $write$ procedure, 
we now have an additional five reorderable pairs of instructions: the first store to $c$ can reorder with the stores to $x1$ and $x2$, the store to $x1$ can reorder with the store to $x2$ and the second store to $c$, and the store to $x2$ can reorder with the second store to $c$.


Each of these pairs needs to be checked using (1) to determine if their reordering can introduce new behaviour. For the first occurrence of $[c]:= r0+1$ and $[x1]:= r1$, we have the following.

\[~~~\wpif([c] := r0_0+1~~ [x1]:=r1, Q)\\
= \wpif([c]:=r0_0+1, \wpif([x1]:=r1,Q)) & (sequential composition)\\
= \M \wpif(\M [c]:=r0_0+1,Q[[x1],\Gamma_{x1}\backslash r1,\Gamma_{r1}] \land\\
([c]\mod 2 = 0 \implies [x1] = r1 \land ([x2]=1\iff\Gamma_{x1}=high)))\O\O & (store)\\
= \M Q[[x1],\Gamma_{x1}\backslash r1,\Gamma_{r1}]\,[[c], \Gamma_{c}\backslash r0_0+1, \Gamma_{r0_0}]\land\\
(r0_0+1\mod 2 = 0 \implies [x1] = r1 \land ([x2]=1\iff\Gamma_{x1}=high)) \land\\
r0_0 + 1 \geq [c] \wedge \Gamma_{r0_0}=low \wedge ([c]\mod 2 = 0 \!\implies\! ([x2]=1\!\iff\! \Gamma_{x1}=high))\O &(store)\\
\zbreak
~\\
~~~\wpif([x1]:=r1~~ [c] := r0_0+1, Q)\\
= \wpif([x1]:=r1, \wpif([c] := r0_0+1,Q)) & (sequential composition)\\
= \M \wpif([\M x1]:=r1,Q[[c],\Gamma_{c}\backslash r0_0+1,\Gamma_{r0_0}] \land r0_0+1 \geq [c]\land \Gamma_{r0_0}=low\land\\
([c]\mod 2 = 0 \implies ([x2]=1\iff \Gamma_{x1}=high)) \land\\
(r0_0+1\mod 2 = 0 \implies ([x2]=1\iff \Gamma_{x1}=high))) \O\O & (store)\\
= \M Q[[c], \Gamma_{c}\backslash r0_0+1, \Gamma_{r0_0}]]\,[[x1],\Gamma_{x1}\backslash r1,\Gamma_{r1}]\land r0_0+1 \geq [c]\land\Gamma_{r0_0}=low\land\\
([c]\mod 2 = 0 \implies ([x2]=1\iff \Gamma_{r1}=high)) \land\\
(r0_0+1\mod 2 = 0 \implies ([x2]=1\iff \Gamma_{r1}=high))) \land\\
([c]\mod 2 = 0 \implies [x1] = r1 \land ([x2]=1\iff\Gamma_{x1}=high)))\O & (store)\]

Recall that the guarantee of $write$ maintains $G = ([c] \mod 2=0 \implies ([x2]=1$ $\iff \Gamma_{x1}=high)$. The weakest precondition for the in-order execution requires, in addition to the substitutions in $Q$, that if $r0_0+1$ (the value at $c$ after the first store) is even, then $[x1]=r1$ to ensure that the second store maintains $[x1]'=[x1]$, and $[x2]=1\iff \Gamma_{x1}=high$ to ensure that $G$ holds before it. Additionally, it requires that the first store does not decrease the value at $c$ and that $G$ holds before it.

The weakest precondition of the reordered execution is similar but, since the store to $x1$ occurs first, requires that $[x1]=r1$ based on the initial value at $c$ not the value $r0_0+1$.

Hence, the reordered weakest precondition is not implied by the in-order one, and rif fails. This failure highlights the security issue that results from persisting the store to $x1$ before that to $c$: if $r1$ is $high$ but $[x2]=0$, since the previous value at $x1$ was $low$, then we no longer maintain $G$.

To prevent such an occurrence we can insert a $\flush~ c$ instruction immediately after the store to $c$. This prevents the stores being persisted out of order. Note that a $\flush_{opt}$ instruction is not sufficient as it could reorder with the following store to $x1$ which could then reorder with the store to $c$. 

The $\flush~c$ instruction also prevents the initial store to $c$ reordering with that to $x2$. However, the stores to $x1$ and $x2$ can still reorder and we consider that possibility next. 

\[~~~ \wpif([x1] := r1~~ [x2]:=r2, Q)\\
= \wpif([x1]:=r1, \wpif([x2]:=r2,Q)) & (sequential composition)\\
= \wpif(\M [x1]:=r1, Q[[x2],\Gamma_{x2} \backslash r2, \Gamma_{r2}] \land \Gamma_{r2}=low\land\\
(c\mod 2=0 \implies [x2]=r2 \land ([x2]=1 \iff \Gamma_{x1}=high))\O & (store)\\
= \M Q[[x2],\Gamma_{x2} \backslash r2, \Gamma_{r2}]\,[x1,\Gamma_{x1}\backslash r1,\Gamma_{r1}]\land \Gamma_{r2}=low \land\\
([c]\mod 2=0 \implies [x2]=r2 \land [x1]=r1\land ([x2]=1 \iff \Gamma_{r1}=high))\O\ & (store)\\
\zbreak
~\\
~~~\wpif([x2] := r2~~ [x1]:=r1, Q)\\
= \wpif([x2]:=r2, \wpif([x1]:=r1,Q)) & (sequential composition)\\
= \wpif(\M [x2]:=r2, Q[[x1],\Gamma_{x1} \backslash r1, \Gamma_{r1}] \land\\
(c\mod 2=0 \implies [x1]=r1 \land ([x2]=1 \iff \Gamma_{x1}=high))\O & (store)\\
= \M Q[[x1],\Gamma_{x1} \backslash r1, \Gamma_{r1}]\,[x2,\Gamma_{x2}\backslash r2,\Gamma_{r2}]\land \Gamma_{r2}=low \land\\
([c]\mod 2=0 \implies [x1]=r1 \land [x2]=r2\land ([x2]=1 \iff \Gamma_{r1}=high))\O & (store)\]

In this case, the weakest preconditions are identical requiring the value at $r2$ to be low since $\LL(x2)=low$ and that if the value at $c$ is even then the values of $x1$ and $x2$ are not changed and that $G$ holds initially. Hence, persisting the stores out of order does not affect security. 

In addition to persisting after the store to $x2$, the store to $x1$ can be persisted after the final store to $c$. The weakest preconditions for this case, $\wpif([x1]:=r1~~ [c]:=r0_1+1, Q)$ and $\wpif([c]:= r0_1 + 1~~ [x1]:= r1)$ are similar to those calculated above but with $r0_1$ in place of $r0_0$. The former does not imply the latter, again indicating a security issue: in this case, if $[x1]$ and $[x2]$ were $high$ and 1, respectively, and an even value is persisted at $c$ and 0 at $x2$ before a $low$ value at $x1$ then, after a power outage, the $high$ value at $x1$ could be read by a reader thread. 

Similarly, since the store to $x2$ can be persisted after the final store to $c$, an even value could be persisted at $c$ and a $high$ value at $x1$ before 1 is persisted at $x2$. Again this allows a reader thread to read the $high$ value in $x1$. Both occurrences can be prevented by including $\flush~x1$ and $\flush~x2$ instructions before the final store to $c$. The resulting code is shown below. 

\[
~~\M r0 := [c]\\
[c] := r0 + 1\\
\flush~c\\
[x1] := r1\\
[x2] := r2\\
r0 := [c]\\
\flush~ x1\\
\flush ~x2\\
[c] := r0 + 1\\
\return\O\]
Since the only stores that can persist out of order are those to $x1$ and $x2$, this code is secure. Since the $read$ procedure does not include any stores, the above changes are sufficient to prove that seqlock is secure on x86 persistent memory.

\section{Soundness}
\label{sec:sound}

The property that $\wpif$ verifies, when the calculated weakest precondition of a program holds, is {\em value-dependent non-interference\/} as defined in \cite{mur16}. This property states that, given two initial states $s_1$ and $s_2$ which agree on the values of variables which are $low$, after executing a prefix of instructions $t$ of the program on each state, the resulting states will continue to agree on the values of variables which are $low$. In other words, the values of variables which are $high$ have no effect on those that are $low$ (and hence the $high$ values cannot be deduced from observations of the $low$ values). Formally, given a program $c$ with specified pre- and postconditions $P$ and $Q$, respectively, we have

\[P \implies \wpif(c, Q) \Rrightarrow\\
\t1  \all s_1, s_2 \mem P; t \leqslant c\cdot \all s_1' \cdot s_1 \sim s_2 \land s_1 \fun_t s_1' \implies \exists s_2'\cdot s_2 \fun_t s_2' \land s_1'\sim s_2'& (2)\]
where $t \leqslant P$ denotes that $t$ is a prefix of $c$, $s_1\sim s_2$ denotes $s_1$ and $s_2$ agree on $low$ values, and $s_1\fun_t s_1'$ denotes $s_1'$ is reached from $s_1$ by instructions $t$. Note that since the programming language is deterministic, the above property implies that all states reached from $s_2$ by $t$ agree with the $low$ values of $s_1'$. The soundness of the logic with respect to this property has been proven in Isabelle/HOL \cite{win21}.

The soundness of rif with respect to the reordering semantics of Colvin and Smith \cite{col18a,col18b} (see Section~\ref{sec:defn}) has also been established in Isabelle/HOL \cite{cou21,cou23}. This proof establishes the fact that if a property (such as (2) above) holds and the rif check (see (1)) holds for all reorderable pairs of instructions, then the property holds when the program is subject to reordering. More formally, for a program $c$ with rely $\R$ and guarantee $\G$, let ${\cal P}(c)$ be a property of $c$ and ${\cal P}_R(c)$ be the same property when $c$ is subject to the reordering relation $R$. We have

\[\t1 {\cal P}(c) \land (\all (\alpha, \beta) \mem reorderable(c,R)\cdot rif(\alpha, \beta, \R, \G)) \Rrightarrow {\cal P}_R(c) & (3)\]
where $reorderable(c,R)$ denotes the set of pairs of reorderable instructions of program $c$ under reordering relation $R$. As mentioned in Section~\ref{sec:defn}, this set can be calculated using a dataflow analysis. 

For the approach in this paper, we additionally need to prove the soundness of the reordering semantics in the context of persistent memory; that is, that our reordering of stores has the same effect on behaviour as a crash due to a power outage or shutdown. We do this by showing that the behaviours allowed under such reordering are equivalent to those of the Intel-x86 persistent memory model Px86$_{\mbox{\footnotesize sim}}$ of Raad~et~al.~\cite{raa20a} (a model developed in consultation with Intel engineers). 

Px86$_{\mbox{\footnotesize sim}}$ has both an operational and declarative semantics which have been proven equivalent. Here we focus on the former. To model delayed stores in x86, the operational semantics models a store happening in two steps: the first records the store in a store buffer $B(\tau)$, where $\tau$ is the identifier of the thread performing the store, and the second transfers the store from the buffer to volatile memory, modelled by a buffer $PB$ of instructions awaiting to be persisted to persistent memory $M$. The first step is captured by the rule \mbox{\sc M-Write}.

\begin{infrule}
\t3 B(\tau) = b
\derive[\mbox{\sc M-Write}]
M,PB,B ~\stackrel{\tau:[x]:=v}{\relbar\joinrel\relbar\joinrel\relbar\joinrel\relbar\joinrel\longrightarrow}~ M,PB,B[\tau \mapsto b.\lseq x,v\rseq]
\end{infrule}
Forwarding is captured by the load rule ({\sc M-Read}) reading the most recent value from the thread's buffer $B$ when one is present (thus reading from a store which has been delayed). Otherwise, it reads the most recent value from volatile memory $PB$ if a value exists there, otherwise from the persistent memory $M$. The rule for $\mfence$ instructions ({\sc M-MF}) requires that $B$ is empty, ensuring all earlier writes have been propagated to volatile memory before the fence. 

Similar to {\sc M-Write}, the rules {\sc M-FL} and {\sc M-FO} add $\lseq fl, x\rseq$ to $B$ for each $\flush~x$ and $\lseq fo, x\rseq$ to $B$ for each $\flush_{opt}~x$, respectively. To ensure atomicity, however, the rule for a $\RMW$ instruction is only enabled when $B$ is empty, and writes directly to volatile memory to ensure that it is immediately visible to all threads. Letting $\epsilon$ denote the empty buffer this is captured by the rule {\sc M-RMW}.

\begin{infrule}
\t2 b \t3 B(\tau) = \epsilon
\derive[\mbox{\sc M-RMW}]
M,PB,B ~\stackrel{\tau:RMW(b,x,e_1,e_2)}{\relbar\joinrel\relbar\joinrel\relbar\joinrel\relbar\joinrel\relbar\joinrel\relbar\joinrel\relbar\joinrel\relbar\joinrel\relbar\joinrel\relbar\joinrel\longrightarrow}~ M,PB.\lseq x,e_2\rseq,B
\end{infrule}

The second step for a store is captured by the rule {\sc M-BPropW}.

\begin{infrule}
B(\tau)=b_1.\lseq x,v\rseq.b_2\t2 \all y,v' \cdot \lseq y, v'\rseq, \lseq fl, y\rseq\nmem b_1
\derive[\mbox{\sc M-BPropW}]
\t1~~ M,PB,B \longrightarrow M,PB.\lseq x,v\rseq, B[\tau\mapsto b_1.b_2]
\end{infrule}
The second premise in rule {\sc M-BPropW} captures the fact that stores may not leave the buffer before stores or $\flush$ instructions added earlier, but may leave the buffer before $\flush_{opt}$ instructions added earlier. Again there are similar rules for propagating $\flush$ and $\flush_{opt}$ instructions to $PB$.

\begin{infrule}
\t2 B(\tau)=b_1.\lseq fl,x\rseq.b_2\\ 
\all y,v\cdot \all x'\mem X \cdot \lseq y, v\rseq, \lseq fo, x'\rseq, \lseq fl, y\rseq\nmem b_1
\derive[\mbox{\sc M-BPropFL}]
M,PB,B \longrightarrow M,PB.\lseq per, x\rseq, B[\tau\mapsto b_1.b_2]
\end{infrule}

\begin{infrule}
\t2 B(\tau)=b_1.\lseq fo,x\rseq.b_2\\
\t1 \all v\cdot \all x'\mem X \cdot \lseq x', v\rseq, \lseq fl, x'\rseq\nmem b_1
\derive[\mbox{\sc M-BPropFO}]
M,PB,B \longrightarrow M,PB.\lseq per, x\rseq, B[\tau\mapsto b_1.b_2]
\end{infrule}
where $X$ is the set of locations on the same cache line as $x$. Note that a $\flush_{opt}$ must not be propagated before earlier stores and flushes to the same cache line, whereas a $\flush$ cannot be propagated earlier than any store or $\flush$, nor earlier than any $\flush_{opt}$ on the same cache line.

Finally, there are rules for persisting writes, ensuring $\flush$ and $\flush_{opt}$ instructions are taken into account.

\begin{infrule}
PB=PB_1.\lseq x, v\rseq.PB_2 \t1~~~~ \all y, v'\cdot \lseq x, v'\rseq, \lseq per, y\rseq\nmem PB_1
\derive[\mbox{\sc M-PropW}]
\t2 M,PB,B \longrightarrow M[x\mapsto v],PB_1.PB_2, B
\end{infrule}

\begin{infrule}
\t1~~~~ PB=PB_1.\lseq per, x\rseq.PB_2\\
 \all y, v\cdot \all x'\mem X \cdot \lseq x', v\rseq, \lseq per, y\rseq\nmem PB_1
\derive[\mbox{\sc M-PropP}]
\t1 M,PB,B \longrightarrow M,PB_1.PB_2, B
\end{infrule}

Given a state $\sigma$ and an x86 assembly program $c$, let a configuration of the program be a tuple $(\sigma,c',p)$ where $c'$ is a suffix of $c$ of instructions yet to be executed and $p$ is a sequence of delayable instructions (stores and flushes) which occur earlier in $c$ than the suffix and are also yet to be executed, i.e., yet to take effect in volatile memory. We use $\epsilon$ to represent $p$ when it is empty. The configuration captures the progress of the program in terms of what has taken place in volatile memory and is hence visible to other threads. 

For the reordering semantics of Section~\ref{sec:defn}, when an instruction is executed, all earlier instructions in $c'$, if any, are added (in program order) to the end of $p$. These instructions and the executed instruction are removed from $c'$. When one of the instructions in $p$ is executed, it is removed from $p$. 

For Px86$_{\mbox{\footnotesize sim}}$, when an instruction is added to buffer $B$ it is added to the end of $p$ and removed from $c'$. When it is propagated from $B$ to $PB$, i.e., takes effect in volatile memory, it is removed from $p$. That is, $p$ denotes the same sequence of instructions as $B$.

\begin{theorem}
Any configuration $(\sigma,c',p)$ of a program $c$ reachable by the reordering semantics when $R$ is defined as in Table~\ref{tab:pm}, is reachable by the semantics of Px86$_{\mbox{\footnotesize sim}}$, and vice versa.
\end{theorem}
\noindent{\bf Proof by induction:} \\
Base case: Initially, the configuration of $c$ in a state $\sigma_0$ is $(\sigma_0,c,\epsilon)$ under both semantics.\\
Inductive case: Assume that the program is in a configuration $(\sigma,c',p)$ reachable by both semantics. Let $\sigma_\alpha$ denote the state resulting from executing instruction $\alpha$ from state $\sigma$.
\begin{enumerate}[(i)]
\item A {\bf load}, {\bf register update} or {\bf jump} in $c'$ can execute under the reordering semantics, only when each instruction occurring before it in $c'$ is a store, $\flush$ or $\flush_{opt}$. As can be seen in Table~\ref{tab:pm}, these are the only instructions which a later load, update or jump can reorder with. That is, when $c'= c_1~ \beta~ c_2$ where $c_1$ is a sequence of stores and flushes and $\beta$ is the load, update or jump, the configuration can become $(\sigma_{\beta_{\lseq\alpha\rseq}},c_2, p~c_1)$, where $\alpha$ is the latest store to the same location in $c_1$ if any (and $\beta_{\lseq\alpha\rseq}$ is $\beta$ otherwise).%
\vspace{2mm}

Similarly, a load, update or jump in $c'$ can execute under Px86$_{\mbox{\footnotesize sim}}$ only when each instruction occurring before it in $c'$ is a store, $\flush$ or $\flush_{opt}$. Each of these instructions (with expressions evaluated to values) can be placed in order into $B$ (via rules {\sc M-Write}, {\sc M-FL} and {\sc M-FO}). The load, update or jump is then able to execute; a load $\beta$ taking on the value of the most recent store $\alpha$ to the same location in $B$ if any (via rule {\sc M-Read}), i.e., executing as $\beta_{\lseq\alpha\rseq}$. Again, the configuration becomes $(\sigma_{\beta_{\lseq\alpha\rseq}}, c_2, p~ c_1)$.\vspace{2mm}

\item A {\bf store} $\alpha$ can execute under the reordering semantics when it is either in $p$ and each instruction before it is a $\flush_{opt}$, or it is in $c'$ and each instruction before it in $c'$ as well as each instruction in $p$ is a $\flush_{opt}$. This follows as the only instructions which a later store can reorder with are $\flush_{opt}$ instructions (see Table~\ref{tab:pm}). In the first case, configuration $(\sigma,c', p_1~ \alpha~ p_2)$ becomes $(\sigma_\alpha,c', p_1~ p_2)$ and in the second, configuration $(\sigma,c_1~\alpha~c_2, p)$ becomes $(\sigma_\alpha,c_2, p~ c_1)$. \vspace{2mm}

Similarly, a store can execute under Px86$_{\mbox{\footnotesize sim}}$ only under the same configurations. This is because propagation to volatile memory (see rule {\sc M-BPropW}) requires there are no other stores or $\flush$ instructions before it in the buffer $B$, only $\flush_{opt}$ instructions. In the case of configuration $(\sigma,c', p_1~ \alpha~ p_2)$, the store can leave $B$ (and hence $p$) 
resulting in $(\sigma_\alpha,c', p_1~p_2)$. In the case of configuration $(\sigma,c_1~\alpha~c_2, p)$, the store can be added to $B$ by rule {\sc M-Write} after each of the $\flush_{opt}$ instructions preceding it in $c'$ are added to $B$ (via rule {\sc M-FO}) resulting in configuration $(\sigma,c_2, p~c_1~\alpha)$. After this, rule {\sc M-BPropW} can be used to propagate the store to volatile memory (since there are no stores or $\flush$ instructions ahead of it in $B$), resulting in configuration $(\sigma_\alpha,c_2, p~ c_1)$. \vspace{2mm}

\item An {\bf RMW} or {\bf mfence} instruction $\alpha$ cannot be reordered with any instruction. Hence, in the reordering semantics it will only execute in a configuration $(\sigma, \alpha~c',\epsilon)$, i.e., when all previous instructions in the program have taken effect in volatile memory. The configuration becomes $(\sigma_\alpha, c', \epsilon)$. \vspace{2mm}

Similarly, an $\RMW$ or $\mfence$ can only execute in Px86$_{\mbox{\footnotesize sim}}$ when $B$ (and hence $p$) is empty. That is, it can only execute from a configuration $(\sigma, \alpha~c',\epsilon)$, resulting in configuration $(\sigma_\alpha, c',\epsilon)$.\vspace{2mm}

\item A {\bf flush} instruction $\alpha$ can only occur in the reordering semantics when it is in $p$ and each instruction before it is a $\flush_{opt}$ belonging to a different cache line, or it is in $c'$ and each instruction before it in $c'$ as well as each instruction in $p$ is a $\flush_{opt}$ belonging to a different cache line (see Table~\ref{tab:pm}). In the first case, configuration $(\sigma, c', p_1~ \alpha~ p_2)$ becomes $(\sigma, c', p_1~ p_2)$ and in the second, configuration $(\sigma, c_1~\alpha~c_2, p)$ becomes $(\sigma, c_2, p~ c_1)$. \vspace{2mm}

Similarly, a $\flush$ can execute under Px86$_{\mbox{\footnotesize sim}}$ only under the same configurations. This is because propagation to volatile memory (see rule {\sc M-BPropFL}) requires there are no stores, other $\flush$ instructions or $\flush_{opt}$ instructions on the same cache line before it in the buffer $B$, only $\flush_{opt}$ instructions on different cache lines. In the case of configuration $(\sigma, c', p_1~ \alpha~ p_2)$, the $\flush$ can leave $B$ (and hence $p$) 
resulting in configuration $(\sigma, c', p_1~p_2)$. In the case of configuration $(\sigma, c_1~\alpha~c_2, p)$, the $\flush$ can be added to $B$ by rule {\sc M-FL} after each of the instructions preceding it in $c'$ are added to $B$ (via rule {\sc M-FO}) resulting in configuration $(\sigma, c_2, p~c_1~\alpha)$. After this, rule {\sc M-BPropFL} can be used to propagate the $\flush$ to volatile memory, resulting in configuration $(\sigma, c_2, p~ c_1)$. \vspace{2mm}

\item A {\bf flush$_{opt}$} instruction $\alpha$ can only occur in the reordering semantics when it is in $p$ and each instruction before it is a store or a $\flush$ belonging to a different cache line or a $\flush_{opt}$, or it is in $c'$ and each instruction before it in $c'$ as well as each instruction in $p$ is a store or a $\flush$ belonging to a different cache line or a $\flush_{opt}$ (see Table~\ref{tab:pm}). In the first case, configuration $(\sigma, c', p_1~ \alpha~ p_2)$ becomes $(\sigma, c', p_1~ p_2)$ and in the second, configuration $(\sigma, c_1~\alpha~c_2, p)$ becomes $(\sigma, c_2, p~ c_1)$. \vspace{2mm}

Similarly, a $\flush_{opt}$ can execute under Px86$_{\mbox{\footnotesize sim}}$ only under the same configurations. This is because propagation to volatile memory (see rule {\sc M-BPropFO}) requires there are no stores or $\flush$ instructions on the same cache line before it in the buffer $B$, only stores and $\flush$ instructions on different cache lines, and $\flush_{opt}$ instructions. In the case of configuration $(\sigma, c', p_1~ \alpha~ p_2)$, the $\flush_{opt}$ can leave $B$ (and hence $p$) 
resulting in the configuration $(\sigma, c', p_1~p_2)$. In the case of configuration $(\sigma, c_1~\alpha~c_2, p)$, the $\flush_{opt}$ can be added to $B$ by rule {\sc M-FO} after each of the instructions preceding it in $c'$ are added to $B$ (via rules {\sc M-Write} and {\sc M-FL}) resulting in configuration $(\sigma, c_2, p~c_1~\alpha)$. After this, rule {\sc M-BPropFO} can be used to propagate the $\flush_{opt}$ to volatile memory, resulting in configuration $(\sigma, c_2, p~ c_1)$.
\vspace{-4mm}\[~~ & $\QED$\]
\end{enumerate}

The above theorem shows that the order of instructions taking effect on volatile memory are identical in our reordering semantics and that of Px86$_{\mbox{\footnotesize sim}}$, ensuring soundness of our reasoning (using the reordering relation of Table~\ref{tab:pm}). This reasoning does not, however, take into account crash states resulting from a power outage. To show soundness under power outages, we need to consider the propagation of instructions to persistent memory. 

For a state $\pi$ corresponding to persistent memory and program $c$, we let $[\pi,c', q]$ denote a persistent configuration where $c'$ is a suffix of $c$ of instructions yet to be executed, and $q$ is a sequence of stores and flushes which occur earlier in $c$ than the suffix and are also yet to take effect in persistent memory. In Px86$_{\mbox{\footnotesize sim}}$, $q$ corresponds to the sequence of the stores in $PB$ and stores and flushes in $B$ in the order that they appear in $c$. Note that flushes in $PB$ have already taken effect in the sense that they can no longer be reordered
(see rules {\sc M-PropW} and {\sc M-PropP} where both kinds of flushes are represented by $per$). 

\begin{theorem}
Any persistent configuration $[\pi, c',q]$ of a program $c$ reachable by the reordering semantics when $R$ is defined as in Table~\ref{tab:out}, is reachable by the semantics of Px86$_{\mbox{\footnotesize sim}}$, and vice versa.
\end{theorem}
\noindent{\bf Proof by induction:} \\
Base case: Initially, the persistent configuration of $c$ in a state $\pi_0$ is $[\pi_0,c,\epsilon]$ under both semantics.\\
Inductive case: Assume that the program is in a persistent configuration $[\pi,c',q]$ reachable by both semantics. Let $\pi_\alpha$ denote the persistent state resulting from executing instruction $\alpha$ from persistent state $\pi$.

\begin{enumerate}[(i)]
\item {\bf Load}, {\bf register update}, {\bf jump}, {\bf RMW}, {\bf mfence}, {\bf flush} and {\bf flush$_{opt}$} instructions update a persistent configuration $[\pi, c',q]$ when they take effect in volatile memory. Hence, given that the instructions they can reorder before are identical in Tables~\ref{tab:pm} and~\ref{tab:out}, the reasoning for these instructions from Theorem~1 suffices to prove that the same persistent configurations can be reached by the reordering semantics and Px86$_{\mbox{\footnotesize sim}}$. \vspace{2mm}

\item A {\bf store} $\alpha$ is persisted in the reordering semantics when it is either in $q$ and each instruction before it is a store to another location or a $\flush_{opt}$ (see Table~\ref{tab:out}), or it is in $c'$ and each instruction before it in $c'$ as well as each instruction in $q$ is a store to another location or a $\flush_{opt}$. In the first case, configuration $[\pi, c', q_1~ \alpha~ q_2]$ becomes $[\pi_\alpha, c', q_1~ q_2]$ and in the second, configuration $[\pi, c_1~\alpha~c_2, q]$ becomes $[\pi_\alpha, c_2, q~ c_1]$. \vspace{2mm}

Similarly, a store can persist under Px86$_{\mbox{\footnotesize sim}}$ only under the same configurations. If the store is in $B$ it can propagate to $PB$ along with any earlier stores in $B$ (without changing the persistent configuration) provided there are only store and $\flush_{opt}$ instructions earlier in $B$ (see rule {\sc M-BPropW}). If it is in $PB$, it can propagate to persistent memory (see rule {\sc M-PropW}) provided there are no earlier stores to the same location in $PB$ (recall the flushes in $PB$ have already taken effect and are not included in $q$).\vspace{2mm}

In the case of configuration $[\pi, c', q_1~ \alpha~ q_2]$, the store  propagates to $PB$ as described above 
and from $PB$ to persistent memory once 
all earlier flushes are removed from $PB$ using rule {\sc M-PropP}. 
The resulting configuration is $[\pi_\alpha, c', q_1~q_2]$. In the case of configuration $[\pi, c_1~\alpha~c_2, q]$, the store can be added to $B$ by rule {\sc M-Write} after each of the instructions preceding it in $c'$ are added to $B$ (rules {\sc M-Write} and {\sc M-FO}). After this, rule {\sc M-BPropW} can be used to propagate all stores in $B$ to $PB$ (since this propagation can occur before preceding $\flush_{opt}$ instructions in $B$). The resulting configuration is $[\pi, c_2, q~ c_1~\alpha]$. Finally, the store can be persisted with rule {\sc M-PropW}: again all flushes are first removed from $PB$ and then the store is persisted which can occur before preceding stores to different locations. The resulting configuration is $[\pi_\alpha,c_2, q~ c_1]$ as required.
\vspace{-4mm}\[~~ & $\QED$\]
\end{enumerate}


\section{Related work}
\label{sec:rel}

To the best of our knowledge, there is no existing work on information flow security on persistent memory. There has been work, however, on information flow on processor memory models, and much activity on formalising persistent memory semantics and associated verification approaches. 

\subsection{Information flow and weak memory models}

Vaughan and Milstein \cite{vau12} provide a logic for detecting information leaks on TSO which cannot be detected using standard information flow techniques. Mantel et al.\ \cite{man14} show that information flow security on a variety of memory models (TSO, PSO and IBM 370) does not imply information flow security in the absence of a weak memory model, nor vice versa. 

Based on the semantics of Colvin and Smith \cite{col18a,col18b}, Smith et al.\ \cite{smi19,smi21} provide a logic that can provide information flow assurance on {\em any\/} processor memory model. That logic, however, only supports a limited form of compositional reasoning where rely/guarantee properties are restricted to read and write permissions on variables, and these permissions are fixed for the duration of a program's execution. Building on this work, a general approach using full rely/guarantee reasoning is presented by Coughlin and Smith \cite{cou22}. 

All of the above approaches lack the ability to separate the reasoning about information flow from that about the memory model, leading to more complex reasoning than the approach in this paper. However, a recent information flow logic for detecting speculative execution vulnerabilities by Coughlin et al.\ \cite{cou24}, like our approach, uses rif to simplify reasoning.

\subsection{Persistent memory semantics and proof techniques}

Izraelevitz et al.\ \cite{izr16} formalise the notion of {\em durable linearizability\/} for reasoning about crash resilience of concurrent data structures executing on persistent memory. This notion requires that all completed operations of a data structure are persisted (so they remain completed after a crash). A relaxed variant {\em buffered durable linearizability\/} is also presented which requires that completed operations are buffered for persistence, but not necessarily persisted. A proof technique based on refinement of IO-automata is provided for durable linearizability by Derrick et al.\ \cite{der19,der21}. These papers, like the work of Izraelevitz et al., do not consider the effects of the processor's memory model.

Raad and Vafeiadis \cite{raa18}, on the other hand, formalise semantics of a memory model PTSO which incorporates the TSO memory model of x86 with persistent memory. They define buffered durable linearizability for this model. Similarly, Raad et al.\ \cite{raa19} provide a memory model PARMv8 for ARMv8 persistent memory (along with a semantics of software transactions). PTSO is based on the concept of {\em epoch persistency\/} of Pelley et al.\ \cite{pel14} where the execution of each thread is divided into several epochs, separated by persist barriers. A persist barrier ensures that all stores prior to the barrier are persisted to memory before those after the barrier. Such barriers are more coarse-grained than the actual persistent primitives ($\flush$, $\flush_{opt}$ and $\wb$) used in Intel x86 persistent memory. Formal semantics of two memory models based on the actual primitives are provided in Raad et al.\ \cite{raa20a}: Px86$_{\mbox{\footnotesize man}}$ which is faithful to the (under-specified) Intel manual, and Px86$_{\mbox{\footnotesize sim}}$, a simplified model capturing the intention of the Intel engineers. It is the latter semantics which informed our use of rif on x86.

A strengthening of Px86$_{\mbox{\footnotesize sim}}$, called SPx86, is presented by Cho et al.\ \cite{cho21}. They argue that in Px86$_{\mbox{\footnotesize sim}}$ the semantics of $\flush$ operations in the presence of external operations, such as writes to a file, is too weak. Their semantics is {\em view-based\/} \cite{kan17} recording the entire history of stores to memory and allowing different threads to have different views on this history, i.e., threads may be able to read old stores. To fix the problem with $\flush$ instructions, they ensure such instructions are executed {\em synchronously\/}, i.e., they block execution until all pending stores on the associated cache line are persisted (rather than allow execution of reorderable instructions to proceed). The synchronous semantics of $\flush$ instructions also appears in the PTSO$_{syn}$ semantics of Khyzha and Lahav \cite{khy21} who show Px86$_{\mbox{\footnotesize sim}}$ and PTSO$_{syn}$ are equivalent in the absence of external operations. Adopting the synchronous $\flush$ semantics for our approach is easily achieved by redefining the reordering relations of Tables~\ref{tab:pm} and~\ref{tab:out}.


A program logic for persistent memory which is sound with respect to Px86$_{\mbox{\footnotesize sim}}$ is presented by Raad et al.\ \cite{raa20b}. The logic, {\em persistent Owicki-Gries\/} (POG), is based on a combination of Owicki-Gries \cite{owi76} and rely/guarantee \cite{jon83,xu97} reasoning, the latter to allow compositional reasoning. In POG, each memory location $x$ has three values associated with it: $x_v$ is a {\em volatile\/} value written to $x$ during program execution, $x_s$ is a {\em synchronously persisted\/} value that would be seen after a crash if the $\flush$ instructions were incremented synchronously, and $x_p$ is a {\em persistent} value that is seen after a crash. The added complexity, not needed in our approach where the program logic is independent of persistent memory, makes it too difficult to reason about $\flush_{opt}$ instructions. Instead, the authors present an approach to transforming (most) programs that use $\flush_{opt}$ to only use $\flush$ and restrict their logic to only support $\flush$ instructions. 

Bila et al.\ \cite{bil22a} provide an alternative program logic based on views \cite{kan17} called {\sc Pierogi} that does support $\flush_{opt}$ instructions. The logic, which is based on Owicki-Gries reasoning, is proven sound with respect to SPx86 of Cho et al.\ \cite{cho21}.

In other work, Bila et al.\ \cite{bil20} formalise a notion of {\em durable opacity\/} enabling formal reasoning about software transactional memory on persistent memory architectures, along with a proof technique based on refinement of IO-automata. The proof technique is modularised in \cite{bil22b} to separate the proof of opacity (perceived atomicity of transactions) \cite{gue08} from that of the persistent memory effects. As such it allows existing proofs of opacity to be reused to prove durable opacity, in much the same way as our approach allows reuse of existing proofs of information flow. Similarly, recent work by D'Osualdo et al.\ \cite{dos23} provides a proof technique for durable linearizability that, like our approach, is independent of the processor memory model and separates the persistency proof from a standard proof of linearizability (correctness of concurrent data structures) \cite{her90}.

\subsection{Crash recovery}

While orthogonal to the aims of this paper, there has also been work on the formalisation of crash recovery both for persistent memory \cite{cho23} and in general, e.g., \cite{kos16,sur23}. Investigating the application of rif in these contexts would make for interesting future work. 

\section{Conclusion}
\label{sec:con}

We have presented an information flow logic for an unstructured language representing assembly code, and shown how it can be used in combination with an existing technique for weak memory models, rif, to reason about security of programs running on persistent memory architectures. We illustrated our approach for Intel x86 persistent memory using a case study based on Linux's seqlock. 

The rif technique provides us with a proof technique that is applicable to all currently available processor memory models, and allows us to separate the reasoning about information flow, the processor's memory model and the effects of persistent memory. This leads to simpler reasoning and enables us to compare the information flow security of a program in the presence and absence of particular architectural features.   

Future work will focus on automated tool support for the approach and potential application of rif to crash recovery on both persistent memory and more generally. 

\subsection*{Acknowledgement}

The author would like to thank Kirsten Winter for her feedback on an earlier draft of this work. 

\bibliography{references}

@article{cho23,
  author       = {Kyeongmin Cho and
                  Seungmin Jeon and
                  Azalea Raad and
                  Jeehoon Kang},
  title        = {Memento: {A} Framework for Detectable Recoverability in Persistent
                  Memory},
  journal      = {Proc. {ACM} Program. Lang.},
  volume       = {7},
  number       = {{PLDI}},
  pages        = {292--317},
  year         = {2023},
  IGNOREurl          = {https://doi.org/10.1145/3591232},
  doi          = {10.1145/3591232},
  bibsource    = {dblp computer science bibliography, https://dblp.org}
}

@inproceedings{kos16,
  author       = {Eric Koskinen and
                  Junfeng Yang},
  editor       = {Rastislav Bod{\'{\i}}k and
                  Rupak Majumdar},
  title        = {Reducing crash recoverability to reachability},
  booktitle    = {Proceedings of the 43rd Annual {ACM} {SIGPLAN-SIGACT} Symposium on
                  Principles of Programming Languages, {POPL} 2016},
  pages        = {97--108},
  publisher    = {{ACM}}, address = {New York},
  year         = {2016},
  IGNOREurl          = {https://doi.org/10.1145/2837614.2837648},
  doi          = {10.1145/2837614.2837648},
  bibsource    = {dblp computer science bibliography, https://dblp.org}
}

@article{sur23,
  author       = {Milijana Surbatovich and
                  Naomi Spargo and
                  Limin Jia and
                  Brandon Lucia},
  title        = {A Type System for Safe Intermittent Computing},
  journal      = {Proc. {ACM} Program. Lang.},
  volume       = {7},
  number       = {{PLDI}},
  pages        = {736--760},
  year         = {2023},
  IGNOREurl          = {https://doi.org/10.1145/3591250},
  doi          = {10.1145/3591250},
  bibsource    = {dblp computer science bibliography, https://dblp.org}
}

@inproceedings{mur16,
  author    = {Toby C. Murray and
               Robert Sison and
               Edward Pierzchalski and
               Christine Rizkallah},
  title     = {Compositional Verification and Refinement of Concurrent Value-Dependent
               Noninterference},
  booktitle = {{IEEE} 29th Computer Security Foundations Symposium, {CSF} 2016}, 
  pages     = {417--431},
  year      = {2016},
  publisher = {{IEEE} Computer Society}, address = {New York},
     IGNOREurl       = {https://doi.org/10.1109/CSF.2016.36},
  doi       = {10.1109/CSF.2016.36},
  bibsource = {dblp computer science bibliography, https://dblp.org}
}

@inproceedings{van05,
  author       = {Peter Vanbroekhoven and
                  Gerda Janssens and
                  Maurice Bruynooghe and
                  Francky Catthoor},
  editor       = {Kwangkeun Yi},
  title        = {Transformation to Dynamic Single Assignment Using a Simple Data Flow
                  Analysis},
  booktitle    = {Programming Languages and Systems, Third Asian Symposium, {APLAS}
                  2005},
  series       = {Lecture Notes in Computer Science},
  volume       = {3780},
  pages        = {330--346},
  publisher    = {Springer}, address = {Berlin-Heidelberg},
  year         = {2005},
  IGNOREurl          = {https://doi.org/10.1007/11575467\_22},
  doi          = {10.1007/11575467\_22},
  bibsource    = {dblp computer science bibliography, https://dblp.org}
}

@inproceedings{che19,
  author       = {Kevin Cheang and
                  Cameron Rasmussen and
                  Sanjit A. Seshia and
                  Pramod Subramanyan},
  title        = {A Formal Approach to Secure Speculation},
  booktitle    = {32nd {IEEE} Computer Security Foundations Symposium, {CSF} 2019},
  pages        = {288--303},
  publisher    = {{IEEE}}, address = {New York},
  year         = {2019},
  IGNOREurl          = {https://doi.org/10.1109/CSF.2019.00027},
  doi          = {10.1109/CSF.2019.00027},
  bibsource    = {dblp computer science bibliography, https://dblp.org}
}

@article{bal22,
  author       = {Alexandro Baldassin and
                  Jo{\~{a}}o Barreto and
                  Daniel Castro and
                  Paolo Romano},
  title        = {Persistent Memory: {A} Survey of Programming Support and Implementations},
  journal      = {{ACM} Comput. Surv.},
  volume       = {54},
  number       = {7},
  pages        = {152:1--152:37},
  year         = {2022},
  IGNOREurl          = {https://doi.org/10.1145/3465402},
  doi          = {10.1145/3465402},
  bibsource    = {dblp computer science bibliography, https://dblp.org}
}

@article{raa18,
  author       = {Azalea Raad and
                  Viktor Vafeiadis},
  title        = {Persistence semantics for weak memory: integrating epoch persistency
                  with the {TSO} memory model},
  journal      = {Proc. {ACM} Program. Lang.},
  volume       = {2},
  number       = {{OOPSLA}},
  pages        = {137:1--137:27},
  year         = {2018},
  IGNOREurl          = {https://doi.org/10.1145/3276507},
  doi          = {10.1145/3276507},
  bibsource    = {dblp computer science bibliography, https://dblp.org}
}

@article{raa19,
  author       = {Azalea Raad and
                  John Wickerson and
                  Viktor Vafeiadis},
  title        = {Weak persistency semantics from the ground up: formalising the persistency
                  semantics of {ARMv8} and transactional models},
  journal      = {Proc. {ACM} Program. Lang.},
  volume       = {3},
  number       = {{OOPSLA}},
  pages        = {135:1--135:27},
  year         = {2019},
  IGNOREurl          = {https://doi.org/10.1145/3360561},
  doi          = {10.1145/3360561},
  bibsource    = {dblp computer science bibliography, https://dblp.org}
}

@article{raa20a,
  author       = {Azalea Raad and
                  John Wickerson and
                  Gil Neiger and
                  Viktor Vafeiadis},
  title        = {Persistency semantics of the {Intel}-x86 architecture},
  journal      = {Proc. {ACM} Program. Lang.},
  volume       = {4},
  number       = {{POPL}},
  pages        = {11:1--11:31},
  year         = {2020},
  IGNOREurl          = {https://doi.org/10.1145/3371079},
  doi          = {10.1145/3371079},
  bibsource    = {dblp computer science bibliography, https://dblp.org}
}

@article{raa20b,
  author       = {Azalea Raad and
                  Ori Lahav and
                  Viktor Vafeiadis},
  title        = {Persistent {Owicki-Gries} reasoning: a program logic for reasoning about
                  persistent programs on {Intel-x86}},
  journal      = {Proc. {ACM} Program. Lang.},
  volume       = {4},
  number       = {{OOPSLA}},
  pages        = {151:1--151:28},
  year         = {2020},
  IGNOREurl          = {https://doi.org/10.1145/3428219},
  doi          = {10.1145/3428219},
  bibsource    = {dblp computer science bibliography, https://dblp.org}
}

@inproceedings{cho21,
  author       = {Kyeongmin Cho and
                  Sung{-}Hwan Lee and
                  Azalea Raad and
                  Jeehoon Kang},
  editor       = {Stephen N. Freund and
                  Eran Yahav},
  title        = {Revamping hardware persistency models: view-based and axiomatic persistency
                  models for {Intel-x86 and Armv8}},
  booktitle    = {{PLDI} '21: 42nd {ACM} {SIGPLAN} International Conference on Programming
                  Language Design and Implementation},
  pages        = {16--31},
  publisher    = {{ACM}}, address = {New York},
  year         = {2021},
  IGNOREurl          = {https://doi.org/10.1145/3453483.3454027},
  doi          = {10.1145/3453483.3454027},
  bibsource    = {dblp computer science bibliography, https://dblp.org}
}

@inproceedings{bil22a,
  author       = {Eleni Bila and
                  Brijesh Dongol and
                  Ori Lahav and
                  Azalea Raad and
                  John Wickerson},
  editor       = {Ilya Sergey},
  title        = {View-Based {Owicki-Gries} Reasoning for Persistent x86-{TSO}},
  booktitle    = {Programming Languages and Systems - 31st European Symposium on Programming,
                  {ESOP} 2022},
  series       = {Lecture Notes in Computer Science},
  volume       = {13240},
  pages        = {234--261},
  publisher    = {Springer}, address = {Berlin-Heidelberg},
  year         = {2022},
  IGNOREurl          = {https://doi.org/10.1007/978-3-030-99336-8\_9},
  doi          = {10.1007/978-3-030-99336-8_9},
  bibsource    = {dblp computer science bibliography, https://dblp.org}
}

@article{dos23,
  author       = {Emanuele D'Osualdo and
                  Azalea Raad and
                  Viktor Vafeiadis},
  title        = {The Path to Durable Linearizability},
  journal      = {Proc. {ACM} Program. Lang.},
  volume       = {7},
  number       = {{POPL}},
  pages        = {748--774},
  year         = {2023},
  IGNOREurl          = {https://doi.org/10.1145/3571219},
  doi          = {10.1145/3571219},
  bibsource    = {dblp computer science bibliography, https://dblp.org}
}

@inproceedings{der19,
  author       = {John Derrick and
                  Simon Doherty and
                  Brijesh Dongol and
                  Gerhard Schellhorn and
                  Heike Wehrheim},
  editor       = {Maurice H. ter Beek and
                  Annabelle McIver and
                  Jos{\'{e}} N. Oliveira},
  title        = {Verifying Correctness of Persistent Concurrent Data Structures},
  booktitle    = {Formal Methods - The Next 30 Years - Third World Congress, {FM} 2019},
  series       = {Lecture Notes in Computer Science},
  volume       = {11800},
  pages        = {179--195},
  publisher    = {Springer}, address = {Berlin-Heidelberg},
  year         = {2019},
  IGNOREurl          = {https://doi.org/10.1007/978-3-030-30942-8\_12},
  doi          = {10.1007/978-3-030-30942-8_12},
  bibsource    = {dblp computer science bibliography, https://dblp.org}
}

@inproceedings{bil20,
  author       = {Eleni Bila and
                  Simon Doherty and
                  Brijesh Dongol and
                  John Derrick and
                  Gerhard Schellhorn and
                  Heike Wehrheim},
  editor       = {Alexey Gotsman and
                  Ana Sokolova},
  title        = {Defining and Verifying Durable Opacity: Correctness for Persistent
                  Software Transactional Memory},
  booktitle    = {Formal Techniques for Distributed Objects, Components, and Systems
                  - 40th {IFIP} {WG} 6.1 International Conference, {FORTE} 2020},
  series       = {Lecture Notes in Computer Science},
  volume       = {12136},
  pages        = {39--58},
  publisher    = {Springer}, address = {Berlin-Heidelberg},
  year         = {2020},
  IGNOREurl          = {https://doi.org/10.1007/978-3-030-50086-3\_3},
  doi          = {10.1007/978-3-030-50086-3_3},
  bibsource    = {dblp computer science bibliography, https://dblp.org}
}

@article{der21,
  author       = {John Derrick and
                  Simon Doherty and
                  Brijesh Dongol and
                  Gerhard Schellhorn and
                  Heike Wehrheim},
  title        = {Verifying correctness of persistent concurrent data structures: a
                  sound and complete method},
  journal      = {Formal Aspects Comput.},
  volume       = {33},
  number       = {4-5},
  pages        = {547--573},
  year         = {2021},
  IGNOREurl          = {https://doi.org/10.1007/s00165-021-00541-8},
  doi          = {10.1007/S00165-021-00541-8},
  bibsource    = {dblp computer science bibliography, https://dblp.org}
}

@article{bil22b,
  author       = {Eleni Bila and
                  John Derrick and
                  Simon Doherty and
                  Brijesh Dongol and
                  Gerhard Schellhorn and
                  Heike Wehrheim},
  title        = {Modularising Verification Of Durable Opacity},
  journal      = {Log. Methods Comput. Sci.},
  volume       = {18},
  number       = {3},
  year         = {2022},
  IGNOREurl          = {https://doi.org/10.46298/lmcs-18(3:7)2022},
  doi          = {10.46298/LMCS-18(3:7)2022},
  bibsource    = {dblp computer science bibliography, https://dblp.org}
}

@inproceedings{gog82,
  author    = {Joseph A. Goguen and
               Jos{\'{e}} Meseguer},
  title     = {Security Policies and Security Models},
  booktitle = {1982 {IEEE} Symposium on Security and Privacy, 1982},
  pages     = {11--20},
  year      = {1982},
publisher = {{IEEE} Computer Society}, address = {New York},
  IGNOREurl       = {https://doi.org/10.1109/SP.1982.10014},
  doi       = {10.1109/SP.1982.10014},
  bibsource = {dblp computer science bibliography, https://dblp.org}
}

@inproceedings{vau12,
  author    = {Jeffrey A. Vaughan and
               Todd D. Millstein},
  title     = {Secure Information Flow for Concurrent Programs under {Total Store
               Order}},
  booktitle = {25th {IEEE} Computer Security Foundations Symposium, {CSF} 2012}, 
  pages     = {19--29},
  year      = {2012},
  editor    = {Stephen Chong},
  publisher = {{IEEE} Computer Society}, address = {New York},
  IGNOREurl       = {https://doi.org/10.1109/CSF.2012.20},
  doi       = {10.1109/CSF.2012.20},
  bibsource = {dblp computer science bibliography, https://dblp.org}
}

@inproceedings{man14,
  author    = {Heiko Mantel and
               Matthias Perner and
               Jens Sauer},
  title     = {Noninterference under Weak Memory Models},
  booktitle = {{IEEE} 27th Computer Security Foundations Symposium, {CSF} 2014}, 
  pages     = {80--94},
  year      = {2014},
  publisher = {{IEEE} Computer Society}, address = {New York},
  IGNOREurl       = {https://doi.org/10.1109/CSF.2014.14},
  doi       = {10.1109/CSF.2014.14},
  bibsource = {dblp computer science bibliography, https://dblp.org}
}

@book{sor11,
  author    = {Daniel J. Sorin and
               Mark D. Hill and
               David A. Wood},
  title     = {A Primer on Memory Consistency and Cache Coherence},
  series    = {Synthesis Lectures on Computer Architecture},
  publisher = {Morgan {\&} Claypool Publishers}, address = {San Rafael CA},
  year      = {2011},
  IGNOREurl       = {https://doi.org/10.2200/S00346ED1V01Y201104CAC016},
  doi       = {10.2200/S00346ED1V01Y201104CAC016},
  bibsource = {dblp computer science bibliography, https://dblp.org}
}

@inproceedings{col18a,
  author    = {Robert J. Colvin and
               Graeme Smith},
  title     = {A Wide-Spectrum Language for Verification of Programs on Weak Memory
               Models},
  booktitle = {Formal Methods - 22nd International Symposium, {FM} 2018}, 
  pages     = {240--257},
  year      = {2018},
  editor    = {Klaus Havelund and
               Jan Peleska and
               Bill Roscoe and
               Erik P. de Vink},
  series    = {Lecture Notes in Computer Science},
  volume    = {10951},
  publisher = {Springer}, address = {Berlin-Heidelberg},
  IGNOREurl       = {https://doi.org/10.1007/978-3-319-95582-7_14},
  doi       = {10.1007/978-3-319-95582-7_14},
  bibsource = {dblp computer science bibliography, https://dblp.org}
}

@article{col18b,
  author    = {R. J. Colvin and G. Smith},
  title     = {A high-level operational semantics for hardware weak memory models},
  journal   = {CoRR},
  volume    = {abs/1812.00996},
  year      = {2018}
}

@book{dij76,
  author       = {Edsger W. Dijkstra},
  title        = {A Discipline of Programming},
  publisher    = {Prentice-Hall},  address = {Hoboken NJ},
  year         = {1976},
  IGNOREurl          = {https://www.worldcat.org/oclc/01958445},
  isbn         = {013215871X},
  bibsource    = {dblp computer science bibliography, https://dblp.org}
}

@book{dij90,
  author       = {Edsger W. Dijkstra and
                  Carel S. Scholten},
  title        = {Predicate Calculus and Program Semantics},
  series       = {Texts and Monographs in Computer Science},
  publisher    = {Springer}, address = {Berlin-Heidelberg},
  year         = {1990},
  IGNOREurl          = {https://doi.org/10.1007/978-1-4612-3228-5},
  doi          = {10.1007/978-1-4612-3228-5},
  isbn         = {978-3-540-96957-0},
  bibsource    = {dblp computer science bibliography, https://dblp.org}
}

@article{smi21,
  author       = {Graeme Smith and
                  Nicholas Coughlin and
                  Toby Murray},
  title        = {Information-flow control on {ARM} and {POWER} multicore processors},
  journal      = {Formal Methods Syst. Des.},
  volume       = {58},
  number       = {1-2},
  pages        = {251--293},
  year         = {2021},
  IGNOREurl          = {https://doi.org/10.1007/s10703-021-00376-2},
  doi          = {10.1007/S10703-021-00376-2},
  bibsource    = {dblp computer science bibliography, https://dblp.org}
}

@inproceedings{cou21,
  author       = {Nicholas Coughlin and
                  Kirsten Winter and
                  Graeme Smith},
  editor       = {Marieke Huisman and
                  Corina S. Pasareanu and
                  Naijun Zhan},
  title        = {Rely/Guarantee Reasoning for Multicopy Atomic Weak Memory Models},
  booktitle    = {Formal Methods - 24th International Symposium, {FM} 2021},
  series       = {Lecture Notes in Computer Science},
  volume       = {13047},
  pages        = {292--310},
  publisher    = {Springer}, address = {Berlin-Heidelberg},
  year         = {2021},
  IGNOREurl          = {https://doi.org/10.1007/978-3-030-90870-6\_16},
  doi          = {10.1007/978-3-030-90870-6_16},
  bibsource    = {dblp computer science bibliography, https://dblp.org}
}

@article{cou23,
  author       = {Nicholas Coughlin and
                  Kirsten Winter and
                  Graeme Smith},
  title        = {Compositional Reasoning for Non-multicopy Atomic Architectures},
  journal      = {Formal Aspects Comput.},
  volume       = {35},
  number       = {2},
  pages        = {8:1--8:30},
  year         = {2023},
  IGNOREurl          = {https://doi.org/10.1145/3574137},
  doi          = {10.1145/3574137},
  bibsource    = {dblp computer science bibliography, https://dblp.org}
}

@inproceedings{bar05,
  author    = {Michael Barnett and
               K. Rustan M. Leino},
  editor    = {Michael D. Ernst and
               Thomas P. Jensen},
  title     = {Weakest-precondition of unstructured programs},
  booktitle = {Proceedings of the 2005 {ACM} {SIGPLAN-SIGSOFT} Workshop on Program
               Analysis For Software Tools and Engineering, PASTE'05},
  pages     = {82--87},
  publisher = {{ACM}}, address = {New York},
  year      = {2005},
  IGNOREurl       = {https://doi.org/10.1145/1108792.1108813},
  doi       = {10.1145/1108792.1108813},
  bibsource = {dblp computer science bibliography, https://dblp.org}
}

@book{bov05,
  author    = {Daniel P. Bovet and
               Marco Cesati},
  title     = {Understanding the Linux Kernel - from {I/O} ports to process management:
               covers version 2.6 {(3.} ed.)},
  publisher = {O'Reilly}, address = {Sebastopol CA},
  year      = {2005},
  IGNOREurl       = {http://www.oreilly.de/catalog/understandlk/index.html},
  isbn      = {978-0-596-00565-8},
  bibsource = {dblp computer science bibliography, https://dblp.org}
}

@article{xu97,
  author    = {Qiwen Xu and
               Willem P. de Roever and
               Jifeng He},
  title     = {The Rely-Guarantee Method for Verifying Shared Variable Concurrent
               Programs},
  journal   = {Formal Aspects of Computing},
  volume    = {9},
  number    = {2},
  pages     = {149--174},
  year      = {1997},
  IGNOREurl       = {https://doi.org/10.1007/BF01211617},
  doi       = {10.1007/BF01211617},
}

@inproceedings{jon83,
  author       = {Cliff B. Jones},
  editor       = {R. E. A. Mason},
  title        = {Specification and Design of (Parallel) Programs},
  booktitle    = {Information Processing 83, Proceedings of the {IFIP} 9th World Computer
                  Congress},
  pages        = {321--332},
  publisher    = {North-Holland/IFIP}, address = {Amsterdam},
  year         = {1983},
  bibsource    = {dblp computer science bibliography, https://dblp.org}
}

@article{cou22,
  author    = {Nicholas Coughlin and
               Graeme Smith},
  title     = {Compositional noninterference on hardware weak memory models},
  journal   = {Sci. Comput. Program.},
  volume    = {217},
  pages     = {102779},
  year      = {2022},
  IGNOREurl       = {https://doi.org/10.1016/j.scico.2022.102779},
  doi       = {10.1016/j.scico.2022.102779},
  bibsource = {dblp computer science bibliography, https://dblp.org}
}

@article{sab09,
  author    = {Andrei Sabelfeld and
               David Sands},
  title     = {Declassification: Dimensions and principles},
  journal   = {J. Comput. Secur.},
  volume    = {17},
  number    = {5},
  pages     = {517--548},
  year      = {2009},
     IGNOREurl       = {https://doi.org/10.3233/JCS-2009-0352},
  doi       = {10.3233/JCS-2009-0352},
  bibsource = {dblp computer science bibliography, https://dblp.org}
}

@inproceedings{ask07,
  author    = {Aslan Askarov and
               Andrei Sabelfeld},
  editor    = {Michael W. Hicks},
  title     = {Localized delimited release: combining the what and where dimensions
               of information release},
  booktitle = {Proceedings of the 2007 Workshop on Programming Languages and Analysis
               for Security, {PLAS} 2007},
  pages     = {53--60},
  publisher = {{ACM}}, address = {New York},
  year      = {2007},
     IGNOREurl       = {https://doi.org/10.1145/1255329.1255339},
  doi       = {10.1145/1255329.1255339},
  bibsource = {dblp computer science bibliography, https://dblp.org}
}

@inproceedings{mur18,
  author    = {Toby C. Murray and
               Robert Sison and
               Kai Engelhardt},
  title     = {{C{\sc overn}:} {A} Logic for Compositional Verification of Information
               Flow Control},
  booktitle = {2018 {IEEE} European Symposium on Security and Privacy, EuroS{\&}P
               2018}, 
  pages     = {16--30},
  year      = {2018},
  publisher = {{IEEE}}, address = {New York},
     IGNOREurl       = {https://doi.org/10.1109/EuroSP.2018.00010},
  doi       = {10.1109/EuroSP.2018.00010},
  bibsource = {dblp computer science bibliography, https://dblp.org}
}

@inproceedings{smi19,
  author    = {Graeme Smith and
               Nicholas Coughlin and
               Toby Murray},
  title     = {Value-Dependent Information-Flow Security on Weak Memory Models},
  booktitle = {Formal Methods - The Next 30 Years - Third World Congress, {FM} 2019},
  editor    = {Maurice H. ter Beek and
               Annabelle McIver and
               Jos{\'{e}} N. Oliveira},
  series    = {Lecture Notes in Computer Science},
  volume    = {11800},
  publisher = {Springer}, address = {Berlin-Heidelberg},
  pages     = {539--555},
  year      = {2019},
     IGNOREurl       = {https://doi.org/10.1007/978-3-030-30942-8\_32},
  doi       = {10.1007/978-3-030-30942-8_32},
  bibsource = {dblp computer science bibliography, https://dblp.org}
}

@inproceedings{win21,
  author    = {Kirsten Winter and
               Nicholas Coughlin and
               Graeme Smith},
  title     = {Backwards-directed information flow analysis for concurrent programs},
  booktitle = {34th {IEEE} Computer Security Foundations Symposium, {CSF} 2021},
  pages     = {1--16},
  publisher = {{IEEE}}, address = {New York},
  year      = {2021},
     IGNOREurl       = {https://doi.org/10.1109/CSF51468.2021.00017},
  doi       = {10.1109/CSF51468.2021.00017},
  bibsource = {dblp computer science bibliography, https://dblp.org}
}

@inproceedings{smi22,
  author       = {Graeme Smith},
  editor       = {Adri{\'{a}}n Riesco and
                  Min Zhang},
  title        = {Declassification Predicates for Controlled Information Release},
  booktitle    = {Formal Methods and Software Engineering - 23rd International Conference
                  on Formal Engineering Methods, {ICFEM} 2022},
  series       = {Lecture Notes in Computer Science},
  volume       = {13478},
  pages        = {298--315},
  publisher    = {Springer}, address = {Berlin-Heidelberg},
  year         = {2022},
  IGNOREurl          = {https://doi.org/10.1007/978-3-031-17244-1\_18},
  doi          = {10.1007/978-3-031-17244-1_18},
  bibsource    = {dblp computer science bibliography, https://dblp.org}
}

@inproceedings{kan17,
  author       = {Jeehoon Kang and
                  Chung{-}Kil Hur and
                  Ori Lahav and
                  Viktor Vafeiadis and
                  Derek Dreyer},
  editor       = {Giuseppe Castagna and
                  Andrew D. Gordon},
  title        = {A promising semantics for relaxed-memory concurrency},
  booktitle    = {Proceedings of the 44th {ACM} {SIGPLAN} Symposium on Principles of
                  Programming Languages, {POPL} 2017},
  pages        = {175--189},
  publisher    = {{ACM}}, address = {New York},
  year         = {2017},
  IGNOREurl          = {https://doi.org/10.1145/3009837.3009850},
  doi          = {10.1145/3009837.3009850},
  bibsource    = {dblp computer science bibliography, https://dblp.org}
}

@inproceedings{izr16,
  author       = {Joseph Izraelevitz and
                  Hammurabi Mendes and
                  Michael L. Scott},
  editor       = {Cyril Gavoille and
                  David Ilcinkas},
  title        = {Linearizability of Persistent Memory Objects Under a Full-System-Crash
                  Failure Model},
  booktitle    = {Distributed Computing - 30th International Symposium, {DISC} 2016},
  series       = {Lecture Notes in Computer Science},
  volume       = {9888},
  pages        = {313--327},
  publisher    = {Springer}, address = {Berlin-Heidelberg},
  year         = {2016},
  IGNOREurl          = {https://doi.org/10.1007/978-3-662-53426-7\_23},
  doi          = {10.1007/978-3-662-53426-7_23},
  bibsource    = {dblp computer science bibliography, https://dblp.org}
}

@inproceedings{pel14,
  author       = {Steven Pelley and
                  Peter M. Chen and
                  Thomas F. Wenisch},
  title        = {Memory persistency},
  booktitle    = {{ACM/IEEE} 41st International Symposium on Computer Architecture,
                  {ISCA} 2014},
  pages        = {265--276},
  publisher    = {{IEEE} Computer Society}, address = {New York},
  year         = {2014},
  IGNOREurl          = {https://doi.org/10.1109/ISCA.2014.6853222},
  doi          = {10.1109/ISCA.2014.6853222},
  bibsource    = {dblp computer science bibliography, https://dblp.org}
}

@article{khy21,
  author       = {Artem Khyzha and
                  Ori Lahav},
  title        = {Taming x86-{TSO} persistency},
  journal      = {Proc. {ACM} Program. Lang.},
  volume       = {5},
  number       = {{POPL}},
  pages        = {1--29},
  year         = {2021},
  IGNOREurl          = {https://doi.org/10.1145/3434328},
  doi          = {10.1145/3434328},
  bibsource    = {dblp computer science bibliography, https://dblp.org}
}

@inproceedings{cou24,
  author       = {Nicholas Coughlin and
                  Kait Lam and
                  Graeme Smith and
                  Kirsten Winter},
  editor       = {Andr{\'{e}} Platzer and
                  Kristin Yvonne Rozier and
                  Matteo Pradella and
                  Matteo Rossi},
  title        = {Detecting Speculative Execution Vulnerabilities on Weak Memory Models},
  booktitle    = {Formal Methods - 26th International Symposium, {FM} 2024},
  series       = {Lecture Notes in Computer Science},
  volume       = {14933},
  pages        = {482--500},
  publisher    = {Springer}, address = {Berlin-Heidelberg},
  year         = {2024},
  IGNOREurl          = {https://doi.org/10.1007/978-3-031-71162-6\_25},
  doi          = {10.1007/978-3-031-71162-6_25},
  bibsource    = {dblp computer science bibliography, https://dblp.org}
}

@inproceedings{kil73,
 author = {G.~A.~Kildall},
 title = {A unified approach to global program optimization},
 booktitle = {Proc. of POPL },
 year = {1973},
 pages = {194--206},
 publisher = {ACM}, address = {New York},
 numpages = {13},
}

@article{sew10,
  author    = {Peter Sewell and
               Susmit Sarkar and
               Scott Owens and
               Francesco Zappa Nardelli and
               Magnus O. Myreen},
  title     = {x86-{TSO}: a rigorous and usable programmer's model for x86 multiprocessors},
  journal   = {Commun. {ACM}},
  volume    = {53},
  number    = {7},
  pages     = {89--97},
  year      = {2010},
  IGNOREurl       = {http://doi.acm.org/10.1145/1785414.1785443},
  doi       = {10.1145/1785414.1785443},
  bibsource = {dblp computer science bibliography, https://dblp.org}
}

@article{alg14,
  author    = {Jade Alglave and
               Luc Maranget and
               Michael Tautschnig},
  title     = {Herding Cats: Modelling, Simulation, Testing, and Data Mining for
               Weak Memory},
  journal   = {{ACM} Trans. Program. Lang. Syst.},
  volume    = {36},
  number    = {2},
  pages     = {7:1--7:74},
  year      = {2014},
  IGNOREurl       = {http://doi.acm.org/10.1145/2627752},
  doi       = {10.1145/2627752},
  bibsource = {dblp computer science bibliography, https://dblp.org}
}

@inproceedings{col21,
  author       = {Robert J. Colvin},
  editor       = {Radu Calinescu and
                  Corina S. Pasareanu},
  title        = {Parallelized Sequential Composition and Hardware Weak Memory Models},
  booktitle    = {Software Engineering and Formal Methods - 19th International Conference,
                  {SEFM} 2021},
  series       = {Lecture Notes in Computer Science},
  volume       = {13085},
  pages        = {201--221},
  publisher    = {Springer}, address = {Berlin-Heidelberg},
  year         = {2021},
  IGNOREurl          = {https://doi.org/10.1007/978-3-030-92124-8\_12},
  doi          = {10.1007/978-3-030-92124-8_12},
  bibsource    = {dblp computer science bibliography, https://dblp.org}
}

@article{owi76,
  author       = {Susan S. Owicki and
                  David Gries},
  title        = {An Axiomatic Proof Technique for Parallel Programs {I}},
  journal      = {Acta Informatica},
  volume       = {6},
  pages        = {319--340},
  year         = {1976},
  IGNOREurl          = {https://doi.org/10.1007/BF00268134},
  doi          = {10.1007/BF00268134},
  bibsource    = {dblp computer science bibliography, https://dblp.org}
}

@article{her90,
  author       = {Maurice Herlihy and
                  Jeannette M. Wing},
  title        = {Linearizability: {A} Correctness Condition for Concurrent Objects},
  journal      = {{ACM} Trans. Program. Lang. Syst.},
  volume       = {12},
  number       = {3},
  pages        = {463--492},
  year         = {1990},
  IGNOREurl          = {https://doi.org/10.1145/78969.78972},
  doi          = {10.1145/78969.78972},
  bibsource    = {dblp computer science bibliography, https://dblp.org}
}

@inproceedings{gue08,
  author       = {Rachid Guerraoui and
                  Michal Kapalka},
  editor       = {Siddhartha Chatterjee and
                  Michael L. Scott},
  title        = {On the correctness of transactional memory},
  booktitle    = {Proceedings of the 13th {ACM} {SIGPLAN} Symposium on Principles and
                  Practice of Parallel Programming, {PPOPP} 2008},
  pages        = {175--184},
  publisher    = {{ACM}}, address = {New York},
  year         = {2008},
  IGNOREurl          = {https://doi.org/10.1145/1345206.1345233},
  doi          = {10.1145/1345206.1345233},
  bibsource    = {dblp computer science bibliography, https://dblp.org}
}

\end{document}